\documentclass[pdflatex,sn-mathphys-num]{sn-jnl}
\usepackage{amsmath,amssymb,mathtools,bm}
\usepackage{microtype}
\usepackage[utf8]{inputenc}

\usepackage{color}
\usepackage{graphicx}
\usepackage{dsfont}
\usepackage{mathrsfs}
\usepackage{verbatim}
\usepackage{enumerate}
\usepackage{tcolorbox}
\usepackage{enumitem}
\tcbuselibrary{breakable}
\usepackage{lipsum}

\theoremstyle{plain}

\theoremstyle{remark}

\usepackage{graphicx}
\usepackage{epic}\usepackage{eepic}\usepackage{epsfig}\usepackage{amsfonts}
\usepackage{latexsym}\usepackage{color}\usepackage{enumerate}\usepackage{bbm}
\allowdisplaybreaks[4]

\hypersetup{
colorlinks,
linkcolor=blue,
anchorcolor=blue,
citecolor=red
}

\theoremstyle{plain}
\newtheorem{theorem}{Theorem}\newtheorem{corollary}[theorem]{Corollary}\newtheorem{lemma}[theorem]{Lemma}\newtheorem{proposition}[theorem]{Proposition}\theoremstyle{definition}
\theoremstyle{remark}
\newtheorem{remark}{Remark}

\newcommand{\nb}{\nonumber}

\newcommand{\Tr}{\operatorname{Tr}}
\newcommand{\E}{\mathbb E}

\newcommand{\cB}{\mathfrak B}

\newcommand{\norm}[1]{\left\lVert#1\right\rVert}

\providecommand{\claimname}{Claim}
\providecommand{\theoremname}{Theorem}

\makeatother
\providecommand{\claimname}{Claim}
\providecommand{\theoremname}{Theorem}

\begin{document}

\title{Reliability Functions of Quantum Soft Covering and Privacy
Amplification via a Mixed-Order R{\'e}nyi Divergence}

\author{\fnm{Shi-Bing} \sur{Li}\textsuperscript{*}}
\author{\fnm{Hongsen} \sur{Qiu}\textsuperscript{\(\dagger\)}}
\author{\fnm{Xinyu} \sur{Zhang}\textsuperscript{\(\ddagger\)}}
\affil[]{%
	\makebox[0pt][c]{%
		\begin{minipage}{\textwidth}
			\centering
			Institute for Advanced Study in Mathematics,\\
			Harbin Institute of Technology, Harbin, 150001, China
		\end{minipage}%
	}%
}

\abstract{
In this paper, we introduce a novel mixed-order Rényi divergence and investigate its fundamental properties. Using this divergence, we define a family of mixed-order order-two Rényi mutual information and Rényi conditional entropies. We derive exact reliability functions of quantum soft covering and privacy amplification under the sandwiched Rényi divergence of order $\alpha\in[2,\infty)$. The former is jointly characterized by the sandwiched and mixed-order order-two Rényi mutual information quantities, while the latter is characterized by the corresponding conditional entropies. The proofs rely on two main ingredients---the noncommutative Rosenthal inequality and a novel inequality established in this paper. These results provide operational interpretations of the proposed mixed-order Rényi divergence. To the best of our knowledge, this is the first exact characterization of the reliability function for quantum soft covering.
}
\maketitle
\begingroup
\renewcommand{\thefootnote}{\fnsymbol{footnote}}

\footnotetext[1]{\texttt{shibingli10@gmail.com}}
\footnotetext[2]{\texttt{chieughongsen@gmail.com}}
\footnotetext[3]{\texttt{xy.zhang@stu.hit.edu.cn}}

\endgroup
\setcounter{footnote}{0}

\section{Introduction}\label{sec:intro}
Matrix trace functionals constitute a fundamental bridge between operator theory, mathematical physics, and quantum information theory.
Since Lieb's seminal concavity
theorem, their joint convexity and concavity properties have played a central
role in the aforementioned fields.
They underlie among other objects, for example, the Wigner--Yanase--Dyson skew information,
Petz quasi-entropies, quantum relative entropies, and data-processing
inequalities
\cite{lieb1973convex,ando1979concavity,petz1986quasi,jencova2010unified,
hiai2013concavity,carlen2016some,zhang2020wigner}.

A particularly useful general family is the three-parameter trace functional
\begin{equation*}
\Psi_{p,q,s}(A,B,K)
:=
\operatorname{Tr}
\left(
B^{\frac q2}K^*A^pKB^{\frac q2}
\right)^s,
\qquad p,q,s\in\mathbb{R},
\end{equation*}
where $A$ and $B$ are positive semidefinite matrices and $K$ is an
operator. Rooted in Lieb’s celebrated concavity theorem~\cite{lieb1973convex}, and further developed through Ando's convexity theorem and subsequent refinements of Lieb-type trace inequalities
\cite{ando1979concavity,epstein1973remarks,hiai2013concavity,carlen2016some},
the study of this functional provides a foundational mathematical framework for establishing data-processing inequalities for various quantum information measures
\cite{petz1986quasi,frank2013monotonicity,Beigi2013sandwiched,zhang2020wigner}.
These properties are important in quantum information theory. With a suitable normalization and a logarithm, these trace functionals give rise to quantum information measures that satisfy key properties such as the data processing inequality.

Several important quantum R{\'e}nyi divergences arise precisely in this way~\cite{petz1986quasi,audenaert2015alpha,muller2013quantum,wilde2014strong}.
For any $\alpha\geq0$ and quantum states $\rho,\sigma$, the order-$\alpha$ Petz R{\'e}nyi divergence between $\rho$ and $\sigma$ is defined as~\cite{petz1986quasi}
\begin{equation*}
D_{\alpha}^*(\rho,\sigma):=\frac{1}{\alpha-1}\log \Psi_{\alpha,1-\alpha,1}(\rho,\sigma,\mathbbm{1})=\frac{1}{\alpha-1}\log\operatorname{Tr}\rho^\alpha\sigma^{1-\alpha},
\end{equation*}
where $\mathbbm{1}$ is the identity operator.
The Petz R{\'e}nyi divergence and its induced entropy, information,
and capacity quantities have found broad operational applications in
quantum state discrimination, composite quantum hypothesis testing and
correlation detection, classical data compression with quantum side
information, classical--quantum channel coding, and quantum privacy
amplification, device-independent cryptography, the
sphere-packing exponent and the reliability function of classical--quantum channels
\cite{AKCMBMA2007discriminating,Hayashi2007error,ANSV2008asymptotic,
HayashiTomamichel2016correlation,CHDH2021non,Dalai2013lower,CHT2019quantum,HahnTanBrown2024,renes2024tight,LiYang2025}.
The order-$\alpha$ sandwiched R{\'e}nyi divergence between $\rho$ and $\sigma$ is
defined as~\cite{muller2013quantum,wilde2014strong}
\begin{equation*}
D_{\alpha}(\rho,\sigma):=\frac{1}{\alpha-1}\log\Psi_{1,\frac{1-\alpha}{\alpha},\alpha}(\rho,\sigma,\mathbbm{1})=\frac{1}{\alpha-1}\log\operatorname{Tr}
\left(\sigma^{\frac{1-\alpha}{2\alpha}}\rho\sigma^{\frac{1-\alpha}{2\alpha}}\right)^\alpha.
\end{equation*}
The sandwiched R{\'e}nyi divergence and its induced information
quantities have found broad operational applications in quantum
hypothesis testing, classical communication over classical--quantum
channels, classical data compression with quantum side information,
and entanglement-assisted and quantum-feedback-assisted communication,
among others \cite{MosonyiOgawa2015quantum,HayashiTomamichel2016correlation,MosonyiOgawa2017strong,CHDH2021non,GuptaWilde2015multiplicativity,CMW2016strong,LiYao2024,LYH2023tight,LiYao2024reliability,LiYaoChannelSimulation2025,TightAnyShotDecoupling2026}.
They have also become fundamental tools in one-shot quantum information
theory
\cite{tomamichel2016quantum}.
More generally, the $\alpha$-$z$ R{\'e}nyi divergence between $\rho$ and $\sigma$ is defined as
\begin{equation*}
D_{\alpha,z}(\rho\|\sigma)
:=\frac{1}{\alpha-1}\log\Psi_{\frac{\alpha}{z},\frac{1-\alpha}{z},z}(\rho,\sigma,\mathbbm{1})=
\frac{1}{\alpha-1}
\log
\operatorname{Tr}
\left(
\sigma^{\frac{1-\alpha}{2z}}
\rho^{\frac{\alpha}{z}}
\sigma^{\frac{1-\alpha}{2z}}
\right)^z.
\end{equation*}
The $\alpha$-$z$ R{\'e}nyi divergence provides a common framework containing both the Petz and sandwiched
families as special cases. Recently, an operational interpretation of the genuinely
two-parameter $\alpha$-$z$ R{\'e}nyi divergences was established
through large-sample and catalytic relative majorization of pairs
of flat states and certain generalizations thereof
\cite{VerhagenTomamichelHaapasalo2026}. Thus, the theory of matrix trace functionals offers
a unified route from trace functional to R{\'e}nyi divergences and their operational applications.

This perspective also provides a systematic mechanism for discovering new
information quantities. Rather than beginning with a desired operational
formula and introducing an auxiliary expression tailored to its proof, one
may start from a trace functional with a distinguished convexity or concavity
property, construct the corresponding divergence and entropy quantities, and
then seek their information-theoretic meanings. A recent example is the
family of two-parameter R{\'e}nyi conditional entropies introduced by Rubboli,
Goodarzi, and Tomamichel from convex trace functionals
\cite{rubboli2024quantum}. These quantities characterize the strong converse
exponent of quantum privacy amplification under the purified distance~\cite{RubboliTomamichel2026},
showing that previously unexplored trace functionals can give rise to
information measures with exact operational significance.

Motivated by this trace-functional viewpoint, we identify in this paper a new
mixed-order order-two R{\'e}nyi divergence generated by a negative-power Lieb
trace functional. For $\alpha>1$, let $\rho$ be a quantum state and $\sigma$ be a positive semi-definite operator 
satisfying
$\operatorname{supp}\rho\subseteq\operatorname{supp}\sigma$, define
\begin{align}
D_{2}^{(\alpha)}(\rho\|\sigma)
:=\log\Psi_{\frac{1-\alpha}{\alpha},-\frac{1}{\alpha},1}(\sigma,\sigma,\rho)
=\log\operatorname{Tr}\left(\rho\sigma^{\frac{1-\alpha}{\alpha}}\rho\sigma^{-\frac{1}{\alpha}}\right).\label{eq:mixed-order-Q}
\end{align}
We set $D_{2}^{(\alpha)}(\rho\|\sigma)=+\infty$ when the support condition is
not satisfied. The quantity therefore combines two different orders and is, in
general, distinct from the Petz, sandwiched, and $\alpha$-$z$ R{\'e}nyi
divergences.

The functional in~\eqref{eq:mixed-order-Q} is closely related to
\[
\Psi_{-p,-1+p,1}(A,A,K),
\qquad p\in(0,1),
\]
whose joint convexity follows from the theory of negative-power Lieb trace
functionals. In \cite{lieb1973convex}, Lieb established that $\Psi_{-p, -1+p, 1}(A, A, K)$ is jointly convex in $(A,K)$ for $p\in(0,1)$. Moreover, joint convexity can imply a data-processing inequality for the trace functional, see \cite{lindblad1974, lindblad1975}. This implies that $D_2^{(\alpha)}(\rho\|\sigma)$ satisfies the data processing inequality for $\alpha>1$. 

From the mixed-order order-two R{\'e}nyi divergence, the mixed-order order-two R{\'e}nyi mutual information of the C-Q state $\rho_{XE}
=
\sum_{x\in\mathcal X}
P_X(x)
|x\rangle\langle x|
\otimes\rho_E^x$ is defined as
\begin{align*}
I_{2}^{(\alpha)}(X:E)_{\rho_{XE}}
:=D_{2}^{(\alpha)}\bigl(\rho_{XE}\big\|\rho_X\otimes\rho_E\bigr),
\end{align*}
and the mixed-order order-two conditional entropy is defined as
\begin{align*}
H_{2}^{(\alpha)}(X|E)_{\rho_{XE}}
:=-D_{2}^{(\alpha)}\bigl(\rho_{XE}\big\|\mathbbm{1}_X\otimes\rho_E\bigr).
\end{align*}
In this paper, we show that these two quantities, together with the sandwiched Rényi mutual information and conditional entropy, characterize the reliability functions of the quantum soft covering and quantum privacy amplification under the sandwiched R{\'e}nyi divergence with order $\alpha\in[2,\infty)$. Specifically, our results are as follows:
\begin{enumerate}
\item Mixed-order order-two R{\'e}nyi divergence (Sect.~\ref{sec:divergence}). We introduce the mixed-order order-two R{\'e}nyi divergence
$D_{2}^{(\alpha)}(\rho\|\sigma)$ and derive its fundamental
structural properties from the convexity theory of negative-power Lieb trace
functionals. 
\item Quantum soft covering (Sect.~\ref{sec:sc}).
Quantum soft covering concerns the approximation of the average output
state induced by a finite random codebook $\mathcal{C}$ through a C-Q
channel $\mathcal{N}:x\mapsto\rho_E^x$ by the target state $\rho_E$
corresponding to a prescribed input distribution $P_X$
\cite{AhlswedeWinter2002}. It is a fundamental tool in quantum
identification, measurement compression, channel simulation, and
quantum cryptography
\cite{AhlswedeWinter2002,bennettEtAl2014,devetak2005private,
DevetakWinter2003,DevetakWinter2005,Winter2005,CaiWinterYeung2004}.
Consider the asymptotic i.i.d. setting, in which the channel is used
$n$ times and $\mathcal{C}_n$ is an i.i.d. random codebook consisting
of $2^{nR}$ codewords drawn independently according to
$P_X^{\otimes n}$.\footnote{\label{fn:common} Without loss of generality, we assume that $2^R$ is an
integer. Consequently, $2^{nR}$ is an integer for every
$n\in\mathbb{N}$. This convention has no effect on the reliability function.} We measure the
discrepancy between the output state induced by $\mathcal{C}_n$ and
the target state $\rho_E^{\otimes n}$ using the sandwiched
R{\'e}nyi divergence of order $\alpha\in[2,\infty)$. When
$R> I_\alpha(X:E)_{\rho_{XE}}$, the expected divergence,
where the expectation is taken over the random codebook, decays
exponentially with the blocklength. We show that its optimal
exponential decay rate, namely, the reliability function, is
\begin{equation}
\min\left\{
R-I_2^{(\alpha)}(X:E)_{\rho_{XE}},
(\alpha-1)
\bigl(R- I_\alpha(X:E)_{\rho_{XE}}\bigr)
\right\},
\end{equation}
where $ I_\alpha(X:E)_{\rho_{XE}}$ denotes the sandwiched
R{\'e}nyi mutual information defined in~\eqref{def:sand-mi}, and
$I_2^{(\alpha)}(X:E)_{\rho_{XE}}$ denotes the mixed-order
order-two R{\'e}nyi mutual information. See
Theorem~\ref{thm:ref-fun}.
Therefore, the reliability function of quantum soft covering is jointly 
determined by the mixed-order order-two R{\'e}nyi
mutual information and the sandwiched R{\'e}nyi mutual information.
\item Quantum privacy amplification (Sect.~\ref{sec:pa}).
Quantum privacy amplification concerns the extraction of private
randomness from a classical random variable $X$ correlated with a
quantum system $E$, and is a fundamental primitive in quantum
cryptography
\cite{DevetakWinter2005,TSSR2011leftover}. To extract private randomness from $n$ i.i.d.\ copies
of $\rho_{XE}$, we employ a random binning function
\[
F_n:\mathcal{X}^n\longrightarrow\{1,\ldots,2^{nR}\}
\]
as a randomness extractor, where each sequence in $\mathcal{X}^n$ is
assigned independently and uniformly to one of the $2^{nR}$ bins.\hyperref[fn:common]{\textsuperscript{\ref*{fn:common}}} The bin
index $F_n(X^n)$ is taken as the extracted key. We measure its
deviation from a uniform key independent of the adversary's quantum
side information using the sandwiched R{\'e}nyi divergence of order
$\alpha\in[2,\infty)$. When
$R<H_\alpha(X|E)_{\rho_{XE}}$, the expected divergence, where the
expectation is taken over the random binning function, decays
exponentially with the blocklength. We show that its optimal
exponential decay rate, namely, the reliability function, is
\begin{equation}
\min\left\{
H_2^{(\alpha)}(X|E)_{\rho_{XE}}-R,\,
(\alpha-1)
\bigl(H_\alpha(X|E)_{\rho_{XE}}-R\bigr)
\right\},
\end{equation}
where $H_\alpha(X|E)_{\rho_{XE}}$ denotes the sandwiched R{\'e}nyi
conditional entropy defined in~\eqref{def:sand-ce}, and
$H_2^{(\alpha)}(X|E)_{\rho_{XE}}$ denotes the mixed-order order-two
R{\'e}nyi conditional entropy. See Theorem~\ref{thm:one-shot}. 
Hence, the reliability function of quantum privacy amplification is 
jointly determined by the mixed-order order-two R{\'e}nyi
conditional entropy and the sandwiched R{\'e}nyi conditional entropy.
\end{enumerate}

Taken together, these results establish a direct chain from matrix trace
functionals to quantum information measures and, ultimately, to operational
error exponents. Starting from the convexity theory of a negative-power Lieb
trace functional, we obtain a new divergence, its associated mutual
information and conditional entropy, and exact operational significance for 
these quantities in the reliability functions of quantum soft covering and quantum privacy amplification.

\emph{Prior works.} For quantum soft covering, \cite{ChengGao2024} derived an achievable lower bound on the reliability function and a lower bound on the strong-converse exponent under the trace distance, rather than complete characterizations of either exponent. The optimal second-order asymptotics under the same criterion were obtained in \cite{ShenGaoCheng2024}. For quantum privacy amplification under the trace-distance security
criterion, the optimal asymptotic rate was characterized
in~\cite{RennerKonig2005,Renner2005}, and the optimal second-order
asymptotics were obtained in~\cite{ShenGaoCheng2024}. At the
exponential scale, an achievable lower bound on the reliability
function under quantum relative entropy was derived
in~\cite{Hayashi2015}. A converse upper bound was later obtained
in~\cite{LYH2023tight}; it coincides with the corresponding achievable
lower bound only in the high-rate regime above a critical rate.
Strong converse exponents under fidelity and
purified distance security criteria were characterized
in~\cite{LiYao2024,RubboliTomamichel2026}.

The remainder of this paper is organized as follows. In Sect.~\ref{sec:pre} we introduce the
necessary notation, definitions and some basic properties. In Sect.~\ref{sec:divergence} we give the definition of the mixed-order order-two R{\'e}nyi divergence and establish some basic properties.
In Sect.~\ref{sec:math} we introduce some mathematical tools and derive a novel inequality. The proofs are given in Sects.~\ref{sec:sc} and~\ref{sec:pa}, where in Sect.~\ref{sec:sc} we prove the characterization of
the reliability function for quantum soft covering, and in Sect.~\ref{sec:pa} we prove the reliability function for quantum privacy amplification.
Finally, in Sect.~\ref{sec:dis} we conclude the paper with some discussion and open questions.
\section{Preliminaries}\label{sec:pre}

\subsection{Basic Notation}

Let $\mathcal{H}$ be a Hilbert space associated with a finite-dimensional
quantum system. We denote the set of linear operators on $\mathcal{H}$
by $\mathcal{L(H)}$, and the positive semidefinite operators by
$\mathcal{P(H)}$. We use the notation $\mathcal{D(H)}$ to represent
the set of quantum states. We use $\mathbbm{1}_{\mathcal{H}}$ to represent the identity
operator on $\mathcal{H}$. When $\mathcal{H}$ is associated
with a system $A$, the above notations $\mathcal{L(H)},\mathcal{P(H)},\mathcal{D(H)}$
and $\mathbbm{1}_{\mathcal{H}}$ are also written as $\mathcal{L}(A),\mathcal{P}(A),\mathcal{D}(A)$
and $\mathbbm{1}_{A}$, respectively. The dimension of system $A$
is denoted by $|A|$. For a linear operator \(A\) and \(p\in(0,\infty)\), we define
\[
\|A\|_p:=\bigl(\operatorname{Tr}|A|^p\bigr)^{1/p},
\qquad
|A|:=\sqrt{A^\dagger A}.
\]
This defines the Schatten $p$-norm for $p\geq1$ and the
Schatten $p$-quasi-norm for $0<p<1$. For \(p=\infty\), \(\|A\|_\infty\) denotes the operator norm, i.e., the largest singular value of \(A\). A classical-quantum (C-Q) state is a bipartite
state of the form
\begin{equation*}
\rho_{XE}=\sum_{x\in\mathcal{X}}q(x)|x\rangle\langle x|_{X}\otimes\rho_{E}^{x},
\end{equation*}
where $\rho_{E}^{x}\in\mathcal{D}(E)$, and $\{|x\rangle\}$ is an orthonormal basis of $\mathcal{H}_{X}$, and $q(x)$ is a probability
distribution on $\mathcal{X}$. 

For any $\alpha\in(0,1)\cup(1,\infty)$, the order-$\alpha$ fidelity between two quantum states $\rho,\sigma\in\mathcal{D}(\mathcal{H})$ is defined as
\begin{equation*}
Q_{\alpha}(\rho\|\sigma)
:=\mathrm{Tr} \left(\sigma^{\frac{1-\alpha}{2\alpha}}\rho\,\sigma^{\frac{1-\alpha}{2\alpha}}\right)^{\alpha}
=\mathrm{Tr} \left(\rho^{\frac{1}{2}}\sigma^{\frac{1-\alpha}{\alpha}}\rho^{\frac{1}{2}}\right)^{\alpha}
=\left\|\sigma^{\frac{1-\alpha}{2\alpha}}\rho\,\sigma^{\frac{1-\alpha}{2\alpha}}\right\|_{\alpha}^{\alpha},
\end{equation*}
where the second equality follows from the fact that the operators
$\sigma^{\frac{1-\alpha}{2\alpha}}\rho\,\sigma^{\frac{1-\alpha}{2\alpha}}$
and
$\rho^{\frac{1}{2}}\sigma^{\frac{1-\alpha}{\alpha}}\rho^{\frac{1}{2}}$
have the same non-zero eigenvalues.
The order-$\alpha$ sandwiched Rényi divergence~\cite{MDSFT2013on,WWY2014strong} is then defined as
\begin{equation*}
D_\alpha(\rho\|\sigma)
:=\frac{1}{\alpha-1}\log Q_\alpha(\rho\|\sigma).
\end{equation*}
In the limit $\alpha\to1$, it reduces to the quantum relative entropy~\cite{Umegaki1954conditional}
\begin{equation*}
\lim_{\alpha\to1} D_\alpha(\rho\|\sigma)
= D(\rho\|\sigma)
:=\mathrm{Tr}\,\rho(\log \rho - \log \sigma).
\end{equation*}
The sandwiched R{\'e}nyi divergence is additive under tensor products. That is, for any  $\rho,\sigma\in\mathcal{P(A)}$ and $\rho',\sigma'\in\mathcal{P(A')}$
\begin{equation*}
D_\alpha
\left(
\rho\otimes\rho'
\,\middle\|\,
\sigma\otimes\sigma'
\right)
=
D_\alpha(\rho\|\sigma)
+
D_\alpha(\rho'\|\sigma').
\end{equation*}

Let $\rho_{XE}=\sum_{x\in\mathcal{X}} p(x)\,|x\rangle\langle x|_X \otimes \rho_E^x$ be a C-Q state. The order-$\alpha$ sandwiched Rényi conditional entropy~\cite{TBH2014relating} is defined as
\begin{align}\label{def:sand-ce}
H_\alpha(X|E)_{\rho_{XE}}
:=& - D_\alpha(\rho_{XE}\,\|\,\mathbbm{1}_{\mathcal{X}} \otimes \rho_E) \nb\\
=& \frac{1}{1-\alpha}\log
\mathrm{Tr}\sum_{x}
\left(
\rho_E^{\frac{1-\alpha}{2\alpha}}\, p(x)\rho_E^x \,\rho_E^{\frac{1-\alpha}{2\alpha}}
\right)^{\alpha}.
\end{align}
Similarly, the order-$\alpha$ sandwiched Rényi mutual information is defined as
\begin{align}\label{def:sand-mi}
I_\alpha(X:E)_{\rho_{XE}}
:=& D_\alpha(\rho_{XE}\,\|\,\rho_X \otimes \rho_E)\nb \\
=& \frac{1}{\alpha-1}\log
\mathrm{Tr}\sum_{x} p(x)
\left(
\rho_E^{\frac{1-\alpha}{2\alpha}}\,\rho_E^x\,\rho_E^{\frac{1-\alpha}{2\alpha}}
\right)^{\alpha}.
\end{align}
In the limit $\alpha \to 1$, these quantities recover the standard von Neumann conditional entropy and mutual information, i.e.,
\begin{align*}
H(X|E)_{\rho_{XE}}
:=& H(\rho_{XE}) - H(\rho_E), \\
I(X:E)_{\rho_{XE}}
:=& D(\rho_{XE}\,\|\,\rho_X \otimes \rho_E)
= H(\rho_X) + H(\rho_E) - H(\rho_{XE}),
\end{align*}
where $H(\rho) := -\mathrm{Tr}\,\rho \log \rho$ is the von Neumann entropy.

Throughout this paper, the
functions $\log$ and $\exp$ are with base $2$, and $\ln$ is with
base $e$.

\section{Mixed-order $(2,\alpha)$ R{\'e}nyi Divergence}\label{sec:divergence}
In this section, we investigate the basic properties of the introduced mixed-order order-two R\'enyi divergence.

The appearance of such a quantity in quantum information theory is new, but its underlying functional form has appeared in the study of quantum Markov semigroups. It is the \emph{quantum $\chi^2$-divergence} in \cite{chi-diver}. Besides, it is also the \emph{$s$-weighted inner product} in \cite{CM17} and plays the role of a divergence measure there.

In the following proposition, we derive the additivity of the mixed-order order-two R{\'e}nyi divergence.
\begin{proposition}[Additivity] \label{prop:multiplicativity}
Let $\rho_1$ and $\rho_2$ be quantum states on $\mathcal H_1$ and
$\mathcal H_2$, respectively, and let $\sigma_1$ and $\sigma_2$ be
positive semi-definite operators on the corresponding Hilbert spaces.
For the product states $\rho_{12} = \rho_1 \otimes \rho_2$ and $\sigma_{12} = \sigma_1 \otimes \sigma_2$, we have
\begin{equation*}
D_2^{(\alpha)}(\rho_{12}\|\sigma_{12}) = D_2^{(\alpha)}(\rho_1\|\sigma_1) + D_2^{(\alpha)}(\rho_2\|\sigma_2).
\end{equation*}
\end{proposition}

\begin{proof}
By the definition of the mixed-order $(2,\alpha)$ R{\'e}nyi relative entropy, we have
\begin{equation*}
D_2^{(\alpha)}(\rho_{12}\|\sigma_{12}) = \log \operatorname{Tr}\left( \rho_{12} \sigma_{12}^{\frac{1-\alpha}{\alpha}} \rho_{12} \sigma_{12}^{-\frac{1}{\alpha}} \right).
\end{equation*}
Using the property $(A \otimes B)^\beta = A^\beta \otimes B^\beta$, the internal terms can be written as
\begin{align*}
\sigma_{12}^{\frac{1-\alpha}{\alpha}} = \sigma_1^{\frac{1-\alpha}{\alpha}} \otimes \sigma_2^{\frac{1-\alpha}{\alpha}}, \,
\sigma_{12}^{-\frac{1}{\alpha}} = \sigma_1^{-\frac{1}{\alpha}} \otimes \sigma_2^{-\frac{1}{\alpha}}.
\end{align*}
Substituting these into the trace and using the operator product rule $(A_1 \otimes A_2)(B_1 \otimes B_2) = (A_1B_1) \otimes (A_2B_2)$, we obtain
\begin{align*}
\rho_{12} \sigma_{12}^{\frac{1-\alpha}{\alpha}} \rho_{12} \sigma_{12}^{-\frac{1}{\alpha}}
&= (\rho_1 \otimes \rho_2) \left( \sigma_1^{\frac{1-\alpha}{\alpha}} \otimes \sigma_2^{\frac{1-\alpha}{\alpha}} \right) (\rho_1 \otimes \rho_2) \left( \sigma_1^{-\frac{1}{\alpha}} \otimes \sigma_2^{-\frac{1}{\alpha}} \right) \nonumber\\
&= \left( \rho_1 \sigma_1^{\frac{1-\alpha}{\alpha}} \rho_1 \sigma_1^{-\frac{1}{\alpha}} \right) \otimes \left( \rho_2 \sigma_2^{\frac{1-\alpha}{\alpha}} \rho_2 \sigma_2^{-\frac{1}{\alpha}} \right).
\end{align*}
Since the trace of a tensor product is the product of the individual traces, i.e., $\operatorname{Tr}(A \otimes B) = \operatorname{Tr}(A)\operatorname{Tr}(B)$, it follows that
\begin{align*}
D_2^{(\alpha)}(\rho_{12}\|\sigma_{12})
&= \log \left[ \operatorname{Tr}\left( \rho_1 \sigma_1^{\frac{1-\alpha}{\alpha}} \rho_1 \sigma_1^{-\frac{1}{\alpha}} \right) \times \operatorname{Tr}\left( \rho_2 \sigma_2^{\frac{1-\alpha}{\alpha}} \rho_2 \sigma_2^{-\frac{1}{\alpha}} \right) \right] \nonumber\\
&= \log \operatorname{Tr}\left( \rho_1 \sigma_1^{\frac{1-\alpha}{\alpha}} \rho_1 \sigma_1^{-\frac{1}{\alpha}} \right) + \log \operatorname{Tr}\left( \rho_2 \sigma_2^{\frac{1-\alpha}{\alpha}} \rho_2 \sigma_2^{-\frac{1}{\alpha}} \right) \nonumber\\
&= D_2^{(\alpha)}(\rho_1\|\sigma_1) + D_2^{(\alpha)}(\rho_2\|\sigma_2).
\end{align*}
This completes the proof.
\end{proof}

As mentioned in the Introduction, Lieb's joint convexity theorem
\cite{lieb1973convex} implies that the trace functional defining
\(D_2^{(\alpha)}(\rho\Vert\sigma)\) is jointly convex in
\((\rho,\sigma)\) for every \(\alpha>1\). By the standard argument of
Lindblad~\cite{lindblad1974,lindblad1975}, this joint convexity,
together with unitary invariance, implies monotonicity under
completely positive trace-preserving maps. Consequently,
\(D_2^{(\alpha)}\) satisfies the data-processing inequality for every
\(\alpha>1\). We omit the standard proof.
\begin{proposition}[Data-processing inequality]
\label{prop:D2-alpha-DPI}
Let \(\alpha>1\), and let
\(\rho,\sigma\in\mathcal{P}(\mathcal{H}_A)\) satisfy the support
condition required for \(D_2^{(\alpha)}(\rho\Vert\sigma)\) to be
finite. Then, for every completely positive trace-preserving map
\(\mathcal{N}:\mathcal{L}(\mathcal{H}_A)
\rightarrow\mathcal{L}(\mathcal{H}_B)\),
\begin{equation}
D_2^{(\alpha)}
\bigl(\mathcal{N}(\rho)\Vert\mathcal{N}(\sigma)\bigr)
\leq
D_2^{(\alpha)}(\rho\Vert\sigma).
\label{eq:D2-alpha-DPI}
\end{equation}
\end{proposition}

In the following two propositions, we compare the mixed-order order-two R{\'e}nyi divergence with the sandwiched R{\'e}nyi divergence.
\begin{proposition} \label{prop:ordering}
Let $\alpha\in[2,\infty)$. For any quantum state $\rho\in\mathcal{D(H)}$ and positive semi-definite operator $\sigma\in\mathcal{P(H)}$, we have
\begin{equation}\label{eq:14:32}
D_2^{(\alpha)}(\rho\|\sigma) \leq D_\alpha(\rho\|\sigma).
\end{equation}
In particular, $D_2^{(2)}(\rho\|\sigma) = D_2(\rho\|\sigma)$.
\end{proposition}

\begin{proof}
From the definition of $D_2^{(2)}(\rho\|\sigma)$, it is easy to check that $D_2^{(2)}(\rho\|\sigma) = D_2(\rho\|\sigma)$. Next, we prove~\eqref{eq:14:32} for $\alpha\in(2,\infty)$.
For notational convenience, set
\begin{equation*}
X:=
\sigma^{\frac{1-\alpha}{2\alpha}}
\rho
\sigma^{\frac{1-\alpha}{2\alpha}} .
\end{equation*}
Then
\begin{align*}
Q_\alpha(\rho\|\sigma)
=
\operatorname{Tr}X^\alpha,\qquad D_2^{(\alpha)}(\rho\|\sigma)
=\log \operatorname{Tr}
X^2\sigma^{\frac{\alpha-2}{\alpha}} .
\end{align*}
By cyclicity of the trace, we may rewrite
\begin{align*}
\operatorname{Tr}
X^2\sigma^{\frac{\alpha-2}{\alpha}}
=
\operatorname{Tr}
X^{\frac{\alpha}{\alpha-1}}
\left(
X^{\frac{\alpha-2}{2(\alpha-1)}}
\sigma^{\frac{\alpha-2}{\alpha}}
X^{\frac{\alpha-2}{2(\alpha-1)}}
\right).
\end{align*}
Applying the H{\"o}lder inequality with conjugate exponents
$\alpha-1$ and $\frac{\alpha-1}{\alpha-2}$, we obtain
\begin{align}\label{618:10}
&\operatorname{Tr}
X^2\sigma^{\frac{\alpha-2}{\alpha}} \nonumber\\
\le\;&
\left\|
X^{\frac{\alpha}{\alpha-1}}
\right\|_{\alpha-1}
\left\|
X^{\frac{\alpha-2}{2(\alpha-1)}}
\sigma^{\frac{\alpha-2}{\alpha}}
X^{\frac{\alpha-2}{2(\alpha-1)}}
\right\|_{\frac{\alpha-1}{\alpha-2}} \nonumber\\
=\;&
\left(\operatorname{Tr}X^\alpha\right)^{\frac1{\alpha-1}}
\left[
\operatorname{Tr}
\left(
X^{\frac{\alpha-2}{2(\alpha-1)}}
\sigma^{\frac{\alpha-2}{\alpha}}
X^{\frac{\alpha-2}{2(\alpha-1)}}
\right)^{\frac{\alpha-1}{\alpha-2}}
\right]^{\frac{\alpha-2}{\alpha-1}} .
\end{align}
It remains to show that the second factor is at most one. By the
Araki--Lieb--Thirring inequality~\cite{Araki1990inequality}, for positive operators $A,B$ and
$s\ge 1$,
\begin{equation*}
\operatorname{Tr}
\left(
B^{1/2}AB^{1/2}
\right)^s
\le
\operatorname{Tr}
B^{s/2}A^sB^{s/2}.
\end{equation*}
Applying this inequality with
$s=\frac{\alpha-1}{\alpha-2},
A=\sigma^{\frac{\alpha-2}{\alpha}},
B=X^{\frac{\alpha-2}{\alpha-1}}$, one obtains
\begin{align}\label{618:11}
\operatorname{Tr}
\left(
X^{\frac{\alpha-2}{2(\alpha-1)}}
\sigma^{\frac{\alpha-2}{\alpha}}
X^{\frac{\alpha-2}{2(\alpha-1)}}
\right)^{\frac{\alpha-1}{\alpha-2}}
\le&
\operatorname{Tr}
X^{1/2}
\sigma^{\frac{\alpha-1}{\alpha}}
X^{1/2} \nonumber\\
=&
\operatorname{Tr}
X\sigma^{\frac{\alpha-1}{\alpha}} = 1.
\end{align}
Consequently, the combination of~\eqref{618:10} and \eqref{618:11} yields
\begin{equation*}
D_2^{(\alpha)}(\rho\|\sigma)
= \log
\operatorname{Tr}
X^2\sigma^{\frac{\alpha-2}{\alpha}}
\le\log
\left(\operatorname{Tr}X^\alpha\right)^{\frac1{\alpha-1}}=
D_\alpha(\rho\|\sigma).
\end{equation*}
This completes the proof.
\end{proof}

\begin{proposition}
Let $\alpha\in(2,\infty)$. For any quantum state $\rho\in\mathcal{D(H)}$ and positive semi-definite operator $\sigma\in\mathcal{P(H)}$, we have
\begin{equation*}
D_2^{(\alpha)}(\rho\|\sigma)\geq D_2(\rho\|\sigma).
\end{equation*}
Moreover, the equality holds if and only if
$[\rho,\sigma]=0$. When $\alpha=2$,
\begin{equation*}
D_2^{(\alpha)}(\rho\|\sigma)= D_2(\rho\|\sigma).
\end{equation*}
\end{proposition}

\begin{proof}
When $\alpha=2$, the result can be easily checked. Now put $\alpha>2$. Set
\begin{equation*}
F(t):=\operatorname{Tr}\left(\rho\sigma^{-t}\rho\sigma^{t-1}\right).
\end{equation*}
We show that $F(t)$ attains its minimum at $t=1/2$. Let
$\sigma=\sum_i\lambda_i |i\rangle\langle i|$, where $\lambda_i>0$. Then
\begin{equation*}
F(t)=\sum_{i,j}|\langle i|\rho|j\rangle|^2
\lambda_i^{-t}\lambda_j^{t-1}.
\end{equation*}
By symmetrizing the indices $i$ and $j$, we obtain
\begin{align*}
F(t)
&=\frac{1}{2}\sum_{i,j}|\langle i|\rho|j\rangle|^2
\left(
\lambda_i^{-t}\lambda_j^{t-1}
+
\lambda_j^{-t}\lambda_i^{t-1}
\right)  \\
&=\sum_{i,j}|\langle i|\rho|j\rangle|^2
(\lambda_i\lambda_j)^{-1/2}
\cosh\left(
\left(t-\frac{1}{2}\right)
\log\frac{\lambda_j}{\lambda_i}
\right).
\end{align*}
Since the hyperbolic cosine function $\cosh(x)$ is even and achieves its global minimum uniquely at $x=0$, the function $g_{i,j}(t)$ is symmetric around $t=\frac{1}{2}$ and attains its global minimum at $t=\frac{1}{2}$. It follows that
\begin{equation*}
F(t)\geq F\left(\frac{1}{2}\right),
\qquad \forall t\in\mathbb{R}.
\end{equation*}
Taking $t=(\alpha-1)/\alpha$, we get
\begin{align*}
\operatorname{Tr}\left(
\rho\sigma^{-\frac{\alpha-1}{\alpha}}
\rho\sigma^{-\frac1\alpha}
\right)
=
F\left(\frac{\alpha-1}{\alpha}\right)
\geq
F\left(\frac{1}{2}\right)
=
\operatorname{Tr}\left(
\rho\sigma^{-\frac{1}{2}}
\rho\sigma^{-\frac{1}{2}}
\right).
\end{align*}
Taking logarithms gives
\begin{equation*}
D_2^{(\alpha)}(\rho\|\sigma)\geq D_2(\rho\|\sigma).
\end{equation*}

It remains to characterize equality. If $[\rho,\sigma]=0$, then $\rho$ has
no matrix elements between distinct eigenspaces of $\sigma$, so only terms with
$\lambda_i=\lambda_j$ survive. Hence equality holds.

Conversely, assume equality holds. Since each summand above is nonnegative after
subtracting its value at $t=1/2$, equality implies that, whenever
$\langle i|\rho|j\rangle\neq0$, we must have
\begin{equation*}
\left(t-\frac{1}{2}\right)
\log\frac{\lambda_j}{\lambda_i}=0.
\end{equation*}
Thus $\lambda_i=\lambda_j$. Therefore $\rho$ has no matrix elements between
distinct eigenspaces of $\sigma$, which is equivalent to
$[\rho,\sigma]=0$.
\end{proof}

With the preceding propositions, we can immediately obtain the following corollary, which establishes the additivity and ordering relations of the mixed-order order-two R{\'e}nyi conditional entropy and mutual information.
\begin{corollary}[Additivity and ordering relations]
\label{cor:additivity-ordering-information}
Let $\alpha\in[2,\infty)$, and let
$\rho_{XE}$ and $\tau_{X'E'}$ be two C-Q states. Then the
mixed-order order-two R{\'e}nyi conditional entropy is additive:
\begin{align*}
&H_2^{(\alpha)}
\left(
XX'\middle|EE'
\right)_{\rho_{XE}\otimes\tau_{X'E'}}
=
H_2^{(\alpha)}(X|E)_{\rho_{XE}}
+
H_2^{(\alpha)}(X'|E')_{\tau_{X'E'}},
\end{align*}
and the mixed-order order-two R{\'e}nyi mutual information is
additive. That is
\begin{align*}
&I_2^{(\alpha)}
\left(
XX':EE'
\right)_{\rho_{XE}\otimes\tau_{X'E'}}
=
I_2^{(\alpha)}(X:E)_{\rho_{XE}}
+
I_2^{(\alpha)}(X':E')_{\tau_{X'E'}}.
\end{align*}
Moreover, the following ordering relations hold.
\begin{align*}
H_\alpha(X|E)_{\rho_{XE}}
\leq
H_2^{(\alpha)}(X|E)_{\rho_{XE}}
\leq
H_2(X|E)_{\rho_{XE}},
\end{align*}
and
\begin{align*}
I_2(X:E)_{\rho_{XE}}
\leq
I_2^{(\alpha)}(X:E)_{\rho_{XE}}
\leq
I_\alpha(X:E)_{\rho_{XE}}.
\end{align*}
In particular, when $\alpha=2$, both inequalities reduce to
equalities.
\end{corollary}

\section{Mathematical tools and novel estimates}\label{sec:math}
This section introduces several mathematical tools that will be used below. We make the convention that, throughout this paper, all matrix-valued random variables are assumed to be discrete
and finitely supported.

The first is the noncommutative Rosenthal inequality, which is established for von Neumann algebras in \cite{JungeXu2008} by Junge and Xu.
\begin{lemma}[Matrix Rosenthal inequality]
\label{lem:matrix-rosenthal}
Let \(Y_1,\ldots,Y_M\) be independent self-adjoint matrix-valued random variables, satisfying
\[
\mathbb E Y_i=0,
\qquad 1\leq i\leq M.
\]
Then for every $p\geq2$, there exists a constant $C_p<\infty$,
depending only on $p$, such that
\begin{align}
\mathbb E_{Y_1,\cdots, Y_M}
\operatorname{Tr}
\left|
\sum_{i=1}^M Y_i
\right|^p
\leq
C_p
\left[
\sum_{i=1}^M
\mathbb E\operatorname{Tr}|Y_i|^p
+
\operatorname{Tr}
\left(
\sum_{i=1}^M\mathbb E Y_i^2
\right)^{p/2}
\right].
\label{eq:matrix-rosenthal}
\end{align}
\end{lemma}

\begin{proof}
If the right-hand side is infinite, there is nothing to prove. We
therefore assume that
\[
\mathbb E\operatorname{Tr}|Y_i|^p<\infty,
\qquad 1\le i\le M.
\]
For a discrete set $\mathcal{X}$, $Y_i$ are random variables on the probability space $(\mathcal{X},\mu_i)$ such that
$$ Y_i\,:\, \mathcal{X}\longrightarrow \mathcal{L}(\mathcal{H})=\mathcal{B}(\mathcal{H}).$$
Here $\mathcal{B}(\mathcal H)$ denotes all the bounded linear operators on $\mathcal{H}$. The equivalence between $\mathcal{L}(\mathcal{H})$ and $\mathcal{B}(\mathcal{H})$ follows from the fact that $\mathcal{H}$ is finite-dimensional. By writing $\mathcal{B}(\mathcal H)$, we emphasize that it is indeed a von Neumann algebra. If we consider the Hilbert space
$$ \mathcal{K}:=L_2(\mathcal{X}^M,\mu; \mathcal{H})\cong L_2(\mathcal{X}^M,\mu)\otimes \mathcal{H},\quad \mu:=\mu_1\times\cdots\times\mu_M, $$
then, for any $g\in \mathcal{K}$, we define
$$(\widetilde{Y_i}g)(x)=Y_i(x_i)g(x),\quad x=(x_1,\cdots,x_M)\in \mathcal{X}^M. $$
Since each $Y_i$ is finitely supported, it is bounded. Therefore, $\widetilde{Y}_i$ is a bounded operator on $\mathcal{K}$. In this sense, we can identify each $Y_i$ as an element in the von Neumann algebra $(\mathcal{M},\tau)$ where
\[
\mathcal M
:=
L_\infty(\mathcal{X}^M,\mu)\,\overline{\otimes}\,\mathcal{B}(\mathcal{H}),
\qquad \tau(\cdot)
:=
\mathbb E\Tr (\cdot).
\]
Here $\overline{\otimes}$ means the von Neumann algebra tensor product. Moreover, the \(\tau\)-preserving conditional expectation
\(\mathcal E:\mathcal M\to \mathcal{B}(\mathcal{H})\) is given by
\[\mathcal{E}(F)=\int_{\mathcal{X}^M}F(x)\,d\mu(x),\quad F\in(\mathcal{M},\tau)\]
and one can check that
\[
\mathcal E(Y_i)=\mathbb E Y_i,\quad 1\leq m\leq M.
\]

The classical independence of \(Y_1,\ldots,Y_M\) implies that
\(Y_1,\cdots, Y_M\) are independent with respect to \(\mathcal E\); see
\cite[Example~1.3]{JungeXu2008}. By the assumption, we also have
\[
\mathcal E(Y_i)=\mathbb E Y_i=0,\quad  \tau(|Y_i|^p)= \mathbb E \Tr |Y_i|^p<\infty,\quad 1\leq i \leq M.
\]
Hence the noncommutative Rosenthal inequality
\cite[Theorem~2.1(i)]{JungeXu2008} applies to
\(Y_1,\ldots,Y_M\) as elements of \(L_p(\mathcal M,\tau)\), which gives
\[
\left(
\mathbb E_{Y_1,\cdots, Y_M}\Tr
\left|\sum_{i=1}^M Y_i\right|^p
\right)^{1/p}
\leq
K_p
\left[
\left(
\sum_{i=1}^M\mathbb E\Tr|Y_i|^p
\right)^{1/p}
+
\left(
\Tr
\left(
\sum_{i=1}^M\mathbb EY_i^2
\right)^{p/2}
\right)^{1/p}
\right],
\]
where the constant \(K_p\) depends only on \(p\). Note that, since each \(Y_i\) is
self-adjoint, the conditional row and column square-function
norms appear in \cite[Theorem~2.1(i)]{JungeXu2008} coincide in the present case.

Then using the scalar inequality
\[
(a+b)^p\leq 2^{p-1}(a^p+b^p),\quad p\geq 2,
\]
one obtains the claimed inequality \eqref{eq:matrix-rosenthal} with \(C_p=2^{p-1}K_p^p\).
\end{proof}

\begin{remark}
We use the standard noncommutative $L_p$-space $L_p(\mathcal M,\tau)$ associated with the semifinite von Neumann algebra
$(\mathcal M,\tau)$; see, for instance, \cite{PX03} and \cite{Xu07} for the standard construction and basic properties.
\end{remark}

Another ingredient in establishing the reliability function is the following novel estimate.
\begin{theorem}
\label{lem:add-change-base}
Let $A,B\in \mathcal{P}(\mathcal{H})$ and set $H=A-B$.  For every real number $\alpha\in(2,\infty)$, there is a
constant $c_\alpha<\infty$, depending only on $\alpha$, such that
\begin{equation*}
 \Tr HA^{\alpha-2}H
 \leq c_\alpha\left(
 \Tr HB^{\alpha-2}H+\Tr|H|^\alpha
 \right).
\end{equation*}
\end{theorem}

\begin{proof}
For convenience, we denote
\begin{equation*}
 \beta=\frac{\alpha-2}{2},
 \quad
 h=\norm{H}_\alpha,
 \quad
 X=\norm{HA^\beta}_2,
 \quad
 Y=\norm{HB^\beta}_2,
 \quad
 Z=h^{\beta+1}.
\end{equation*}
Thus $\alpha=2(\beta+1)$ and
\begin{equation}\label{eq:add-XYZ-squares}
 X^2=\Tr HA^{\alpha-2}H,
 \qquad
 Y^2=\Tr HB^{\alpha-2}H,
 \qquad
 Z^2=\Tr|H|^\alpha.
\end{equation}
Our purpose is to prove $X^2 \leq c_\alpha(Y^2+Z^2)$.

For
$0\leq t\leq\beta$, define $q_t$ by
\begin{equation*}
 \frac1{q_t}
 =\frac{1-t/\beta}{\alpha}+\frac{t/\beta}{2}=\frac{t+1}{\alpha}.
\end{equation*}
Define $F(z)=HA^{\beta z}$. Thus for $y\in \mathbb{R}$,
$$ \Vert F(iy) \Vert_\alpha=\Vert HA^{i\beta y} \Vert_\alpha=\Vert H\Vert_\alpha=h,\quad \Vert F(1+iy) \Vert_2=\Vert HA^{\beta(1+i y)} \Vert_2=\Vert HA^{\beta}\Vert_2=X. $$
Complex interpolation (cf. Lemma~\ref{lem:complex-interpolation-schatten}) then gives
\begin{equation}\label{eq:add-interpolation-A}
 \norm{HA^t}_{q_t}
 \leq h^{1-t/\beta}X^{t/\beta}.
\end{equation}
Applying the same argument on $B$, one obtains
\begin{equation}\label{eq:add-interpolation-B}
 \norm{HB^t}_{q_t}
 \leq h^{1-t/\beta}Y^{t/\beta}.
\end{equation}
For singular $A$ or $B$, these estimates follow by replacing the operator
with $A+\varepsilon I$ or $B+\varepsilon I$ and then letting
$\varepsilon\downarrow0$.

Write $\beta=m+\theta$, where $m=\lfloor\beta\rfloor$ and
$0\leq\theta<1$.  We use the decomposition
\begin{align}
 H(A^\beta-B^\beta)
 &=HA^m(A^\theta-B^\theta)+H(A^m-B^m)B^\theta\notag\\
 &=HA^m(A^\theta-B^\theta)
   +\sum_{k=0}^{m-1}HA^kHB^{m-1-k+\theta}.
 \label{eq:add-power-decomposition}
\end{align}
When $m=0$ the sum is empty, and when $\theta=0$ the first term vanishes. We first consider the case when $ m\theta\neq 0$. 

The Birman--Koplienko--Solomyak inequality (see Lemma~\ref{BKS}, which is restated from
\cite[Theorem~3]{BKS75}; see also
\cite[Proposition~4.9]{BS89}) gives, for $0<\theta<1$,
\begin{equation}\label{eq:add-BKS}
 \norm{A^\theta-B^\theta}_{\alpha/\theta}
 \leq C_\theta\norm{A-B}_\alpha^\theta
 =C_\theta h^\theta.
\end{equation}
Combining \eqref{eq:add-interpolation-A} and
\eqref{eq:add-BKS}, applying H\"older's inequality with
\[
 \frac12=\frac1{q_m}+\frac{\theta}{\alpha}, \quad 1-\frac{m}{\beta}+\theta=\theta(\frac{1}{\beta}+1),
\]
we obtain
\begin{equation*}
 \norm{HA^m(A^\theta-B^\theta)}_2
 \leq \norm{HA^m}_{q_m}\norm{A^\theta-B^\theta}_{\alpha/\theta} \leq C_\theta h^{1-m/\beta}X^{m/\beta}h^{\theta}=C_\theta X^{m/\beta}Z^{\theta/\beta}.
\end{equation*}
For each term in the sum in \eqref{eq:add-power-decomposition}, put
$\ell_k=\beta-1-k=m-1-k+\theta$.  Then
\eqref{eq:add-interpolation-A} and \eqref{eq:add-interpolation-B}, together
with
\[
 \frac1{q_k}+\frac1{q_{\ell_k}}=\frac{1+k}{\alpha}+\frac{1+\ell_k}{\alpha}=\frac12,
\]
give
\begin{align}
 \norm{HA^kHB^{\ell_k}}_2\leq\norm{HA^k}_{q_k}\norm{HB^{\ell_k}}_{q_{\ell_k}}\leq
 X^{k/\beta}Y^{(\beta-1-k)/\beta}Z^{1/\beta}.
 \label{eq:add-integer-terms}
\end{align}
Here the three exponents on the right-hand side add up to $1$.

Since $HA^\beta=HB^\beta+H(A^\beta-B^\beta)$, we get 
$$ X=\Vert HA^{\beta}\Vert_2\leq \Vert HB^\beta\Vert_2+\Vert H(A^\beta-B^\beta)\Vert_2.$$
Then~\eqref{eq:add-power-decomposition}--\eqref{eq:add-integer-terms}
yield
\begin{align*}
 X\leq Y+C_\theta X^{m/\beta}Z^{\theta/\beta}+\sum_{k=0}^{m-1}
 X^{k/\beta}Y^{(\beta-1-k)/\beta}Z^{1/\beta},
\end{align*}
Note that
$$ X^{k/\beta}Y^{(\beta-1-k)/\beta}Z^{1/\beta}\leq X^{k/\beta}(Y+Z)^{1-k/\beta}.$$
Weighted Young's
inequality (cf. Lemma~\ref{w-young}) then gives, for each $k$ and $\varepsilon_k>0$,
\begin{equation*}
  X^{k/\beta}Y^{(\beta-1-k)/\beta}Z^{1/\beta}\leq\varepsilon_k X+C_{k,\varepsilon_k}(Y+Z),\quad  X^{m/\beta}Z^{\theta/\beta}\leq\varepsilon_m X+C_{m,\varepsilon_m}Z.
\end{equation*}
Choose $\varepsilon_0=\cdots=\varepsilon_{m-1}=C_{\theta}\varepsilon_m=\frac{1}{2(m+1)}$ . Then, we can take
$$ K_\alpha=C_{0,\varepsilon_0}+C_{1,\varepsilon_1}+\cdots+C_{m-1,\varepsilon_{m-1}}+C_\theta C_{m,\varepsilon_m}$$
 such that
$$ X\leq \frac{1}{2}X+K_\alpha (Y+Z).$$
Here $K_\alpha$ only depends on $\alpha$ because $(\beta,m,\theta)$ is uniquely determined by $\alpha$. Taking $c_\alpha=8K_\alpha^2$, this further yields
$$ X^2\leq c_\alpha (Y^2+Z^2).$$
Finally, the case $m=0$ ($\theta=0$) is handled in a similar way, with the sum (the fractional term) in \eqref{eq:add-power-decomposition} omitted. With the convention in \eqref{eq:add-XYZ-squares}, this proves the desired result.
\end{proof}

Throughout this paper, for positive semidefinite operators \(A\) and \(B\),
we use the notation
\begin{equation}
\label{eq:add-bregman-def}
\cB_{\alpha}(A,B)
:=
\operatorname{Tr} A^{\alpha}
-\operatorname{Tr} B^{\alpha}
-\alpha\operatorname{Tr}
B^{\alpha-1}(A-B).
\end{equation}
\begin{lemma}\label{lem:add-bregman}
Let $A,B\in \mathcal{P}(\mathcal{H})$ and set $H=A-B$.  For every real number $\alpha\in(2,\infty)$,
\begin{align}
 (\alpha-1)\Tr HB^{\alpha-2}H
 &\leq \cB_\alpha(A,B),
 \label{eq:add-bregman-lower}\\
 \cB_\alpha(A,B)
 &\leq C_\alpha\left(
       \Tr HB^{\alpha-2}H+\Tr|H|^\alpha
       \right),
 \label{eq:add-bregman-upper}
\end{align}
where $C_\alpha$ depends only on $\alpha$.
\end{lemma}

\begin{proof}
Take spectral decompositions
\[
 A=\sum_i a_i|e_i\rangle\langle e_i|,
 \qquad
 B=\sum_j b_j|f_j\rangle\langle f_j|,
\]
and put \(
 u_{ij}=\langle e_i,f_j\rangle,\,
 w_{ij}=|u_{ij}|^2
\).
The matrix $(w_{ij})$ is doubly stochastic.  Direct expansion gives the
two exact identities
\begin{align}
 \cB_\alpha(A,B)
 &=\sum_{i,j}w_{ij}
   \left[a_i^\alpha-b_j^\alpha
   -\alpha b_j^{\alpha-1}(a_i-b_j)\right],
 \label{eq:add-spectral-bregman}\\
 \Tr HB^{\alpha-2}H
 &=\sum_{i,j}w_{ij}(a_i-b_j)^2b_j^{\alpha-2}.
 \label{eq:add-spectral-kms}
\end{align}

For scalars $x,y\geq0$, we apply the following Taylor's integral formula
\begin{align*}
 x^\alpha-y^\alpha-\alpha y^{\alpha-1}(x-y)\notag
 =\alpha(\alpha-1)(x-y)^2
   \int_0^1(1-t)\bigl[y+t(x-y)\bigr]^{\alpha-2}\,dt.
\end{align*}
The quantity $y+t(x-y)$ is at least $(1-t)y$ and
$\int_0^1(1-t)^{\alpha-1}dt=1/\alpha$.  Hence we obtain
\begin{equation}\label{eq:add-scalar-lower}
 x^\alpha-y^\alpha-\alpha y^{\alpha-1}(x-y)
 \geq(\alpha-1)(x-y)^2y^{\alpha-2}.
\end{equation}
Applying \eqref{eq:add-scalar-lower} term by term to \eqref{eq:add-spectral-bregman} and comparing the result with \eqref{eq:add-spectral-kms}, this
proves \eqref{eq:add-bregman-lower}.

For the upper bound, the convexity of $T\mapsto\Tr T^\alpha$ gives
$$\cB_\alpha(A,B)\geq0,\quad \cB_\alpha(B,A)\geq0.$$
 Hence
\begin{align*}
 \cB_\alpha(A,B)\leq \cB_\alpha(A,B)+\cB_\alpha(B,A)=\alpha\sum_{i,j}w_{ij}
   (a_i-b_j)(a_i^{\alpha-1}-b_j^{\alpha-1}).
\end{align*}
For scalars $x,y\geq0$,
\[(x-y)(x^{\alpha-1}-y^{\alpha-1})=(\alpha-1)(x-y)^2\int_0^1\bigl((1-t)y+tx\bigr)^{\alpha-2}\,dt.\]
Since $\bigl((1-t)y+tx\bigr)^{\alpha-2}\leq x^{\alpha-2}+y^{\alpha-2}$, we have
\begin{equation*}
 (x-y)(x^{\alpha-1}-y^{\alpha-1})
 \leq(\alpha-1)(x-y)^2
       (x^{\alpha-2}+y^{\alpha-2}).
\end{equation*}
Consequently,
\begin{equation}\label{eq:add-sym-upper}
  \begin{aligned}
 \cB_\alpha(A,B)
 &\leq \alpha(\alpha-1)\sum_{i,j}w_{ij}
   (a_i-b_j)^2(a_i^{\alpha-2}+b_j^{\alpha-2})\\
 &=\alpha(\alpha-1)
 \left(\Tr HA^{\alpha-2}H+\Tr HB^{\alpha-2}H\right).
  \end{aligned}
\end{equation}
Applying Theorem~\ref{lem:add-change-base} to $\Tr HA^{\alpha-2}H$  in \eqref{eq:add-sym-upper} proves
\eqref{eq:add-bregman-upper}.
\end{proof}

\begin{proposition}\label{b-est}
Let \(\alpha\geq2\) and let \(A,B\geq0\). Then
\begin{equation}
\mathfrak B_\alpha(A,B)
\geq
2^{1-\alpha}\|A-B\|_\alpha^\alpha.
\label{eq:Bregman-uniform-convexity}
\end{equation}
\end{proposition}

\begin{proof}
The Clarkson--McCarthy inequality (cf. Lemma~\ref{lem:CM}) gives
\[
\|X+Y\|_\alpha^\alpha+\|X-Y\|_\alpha^\alpha
\geq
2\left(
\|X\|_\alpha^\alpha+\|Y\|_\alpha^\alpha
\right).
\]
Taking
\[
X=\frac{A+B}{2},
\qquad
Y=\frac{A-B}{2},
\]
we obtain
\begin{equation}
\frac{\operatorname{Tr}|A|^\alpha+\operatorname{Tr}|B|^\alpha}{2}
-\operatorname{Tr}\left|\frac{A+B}{2}\right|^\alpha
\geq
2^{-\alpha}\|A-B\|_\alpha^\alpha.
\label{eq:Clarkson-midpoint-remainder}
\end{equation}
By Klein's inequality (see \cite[Eric Carlen's part, Theorem 2.11]{carlennote}) for the convex function
$t\mapsto t^\alpha$,
\[
\operatorname{Tr}\left(\frac{A+B}{2}\right)^\alpha
\geq
\operatorname{Tr}B^\alpha
+
\frac{\alpha}{2}
\operatorname{Tr}B^{\alpha-1}(A-B).
\]
Combining this with \eqref{eq:Clarkson-midpoint-remainder} proves \eqref{eq:Bregman-uniform-convexity}.
\end{proof}

The following is another estimate needed below.
\begin{lemma}\label{lem:add-positive-bridge}
Let $\mathcal A\geq0$ be a matrix-valued random variable. For every real number $\alpha\geq2$,
\begin{equation*}
 \Tr (\E \mathcal A^2)^{\alpha/2}
 \leq
 \left[\left(\Tr (\E \mathcal A^2)(\E \mathcal A)^{\alpha-2}\right)
 \left(\E\Tr \mathcal A^\alpha\right)\right]^{1/2}.
\end{equation*}
\end{lemma}
\begin{proof}
Write the distribution of $\mathcal A$ as $P(\mathcal{A}=A_x)=w_x$ and put
\[
 \mathscr A=\bigoplus_x A_x,
 \qquad
 J\xi=(\sqrt{w_x}\,\xi)_x,
 \qquad
 T=\mathscr A^{1/2}J.
\]
Then
\begin{equation}\label{eq:add-block-identities}
 T^*T=\E \mathcal A,
 \qquad
 T^*\mathscr A T=\E \mathcal A^2,
 \qquad
 T^*\mathscr A^qT=\E \mathcal A^{q+1}.
\end{equation}
For $Y\geq0$ and $q\geq1$, the
Araki--Lieb--Thirring inequality~\cite{Araki1990inequality} gives
\begin{align*}
 \Tr(T^*YT)^q=\Tr(Y^{1/2}TT^*Y^{1/2})^q\notag\leq\Tr Y^q(TT^*)^q
 =\Tr (T^*T)^{q-1}T^*Y^qT.
\end{align*}
Here we used
$(TT^*)^q=T(T^*T)^{q-1}T^*$ on the support of $TT^*$. In fact, using the polar decomposition $T=U|T|$, we get
\[
T^*T=|T|^2,
\quad
TT^*=U|T|^2U^*,\quad
(TT^*)^q
=
\left(U|T|^2U^*\right)^q
=
U|T|^{2q}U^*.
\]
Moreover,
\[
\begin{aligned}
T(T^*T)^{q-1}T^*=U|T|\,(|T|^2)^{q-1}|T|U^*=U|T|\,|T|^{2q-2}|T|U^*=U|T|^{2q}U^*.
\end{aligned}
\]
Hence
\((TT^*)^q=T(T^*T)^{q-1}T^*\).

Put $Y=\mathscr A$ and $q=\alpha/2\geq 1$. By
\eqref{eq:add-block-identities},
\begin{equation}\label{eq:add-alt-applied}
 \Tr (\E \mathcal A^2)^q
 \leq\Tr (\E\mathcal{A}^{q+1})(\E\mathcal{A})^{q-1}=\E \Tr \mathcal{A}^{q+1}(\E\mathcal{A})^{q-1}.
\end{equation}
Cauchy--Schwarz inequality then gives
\begin{align}
 \E\Tr \mathcal A^{q+1}(\E\mathcal{A})^{q-1}
 &=\E\Tr \mathcal A^{q} \mathcal{A}(\E\mathcal{A})^{q-1}\notag\\
 &\leq
 \left(\E\Tr \mathcal A^{2q}\right)^{1/2}
 \left(\E\Tr (\E\mathcal{A})^{q-1}\mathcal A^2(\E\mathcal{A})^{q-1}\right)^{1/2}\notag\\
 &=\left[
   \left(\E\Tr \mathcal A^\alpha\right)
    \left(\Tr (\E \mathcal A^2)(\E \mathcal A)^{\alpha-2}\right)
   \right]^{1/2}.
 \label{eq:add-bridge-cs}
\end{align}
Combining \eqref{eq:add-alt-applied} and \eqref{eq:add-bridge-cs} proves
the desired result.
\end{proof}

\section{Reliability Function of Quantum Soft Covering}\label{sec:sc}

In this section, we derive the reliability function of quantum soft covering for i.i.d. random codebooks under the sandwiched R{\'e}nyi divergence for orders $\alpha \in [2,\infty)$. To the best of our knowledge, this is the first work to provide an exact characterization of the reliability function in the quantum soft covering setting.
\subsection{Problem Statement and Main Results}

Let $\mathcal{N}:\mathcal{X}\to \mathcal{D}(E)$ be a C-Q channel with $\mathcal{N}(x)= \rho_E^x$.
Given an input distribution \(P_X\), the corresponding output is then given by the marginal state
\begin{equation*}
\rho_E=\sum_{x\in\mathcal{X}}P_X(x)\rho_E^x=\mathbb E_X\rho_E^{X}.
\end{equation*}
The goal of quantum soft covering is
to approximate the marginal state at the channel output,
given access to the C-Q channel $\mathcal{N}:x\mapsto\rho_{E}^{x}$
and sampling from the prior distribution $P_{X}$.

To that end, we consider a random codebook $\mathcal{{C}}=\{X(m)\}_{m=1}^{M}$
of size $M$, where its codewords $\{X(1),\cdot\cdot\cdot,X(M)\}$
are independently generated according to $P_{X}$. Then, the average
state induced by the random codebook $\mathcal{{C}}$ is

\begin{equation*}
\rho_{E}^{\mathcal{{C}}}=\frac{1}{M}\sum_{x\in\mathcal{{C}}}\rho_{E}^{x}=\frac{1}{M}\sum_{m=1}^{M}\rho_{E}^{f_{\mathcal{C}}(m)},
\end{equation*}
where we set $f_{\mathcal{C}}(m)=X(m).$
In contrast to previous works on the quantum soft covering, we take
the order-$\alpha$ sandwiched Rényi divergence
\begin{equation*}
D_\alpha(\rho_{\mathcal{C}E}\|\rho_{\mathcal{C}}\otimes\rho_{E}):=\frac{1}{\alpha-1}\log\mathbb{E_{\mathcal{C}}}Q_{\alpha}(\rho_{E}^{\mathcal{C}}\|\rho_{E})
\end{equation*}
to measure the discrepancy between the codebook-induced state $\rho_{E}^{\mathcal{{C}}}$
and the true marginal state $\rho_{E}$.

For the memoryless extension $\mathcal{N}^{\otimes n}: \mathcal{X}^n \to \mathcal{D}(E^{\otimes n})$, the channel acts as
\begin{equation*}
\mathcal{N}^{\otimes n}(x^n)=\rho_{E^n}^{x^n}
=\bigotimes_{i=1}^n \rho_E^{x_i}.
\end{equation*}
Given the product input distribution $P_X^{\otimes n}$, the corresponding output marginal state is
\begin{equation*}
\rho_{E^n}=\sum_{x^n\in\mathcal{X}^n} P_X^{\otimes n}(x^n)\,\rho_{E^n}^{x^n}.
\end{equation*}
We consider a random codebook $\mathcal{C}_n=\{X^n(m)\}_{m=1}^{M=2^{nR}}$,
where each codeword $X^n(m)$ is independently drawn according to $P_X^{\otimes n}$ and the non-negative number $R$ is the codebook rate. The induced average output state is
\begin{equation*}
\rho_{E^n}^{\mathcal{C}_n}
=\frac{1}{M}\sum_{m=1}^{M}\rho_{E^n}^{X^n(m)}.
\end{equation*}
The performance of quantum soft covering in the $n$-shot setting is then characterized by
\begin{equation*}
D_\alpha(\rho_{\mathcal{C}_n E^n}\|\rho_{\mathcal{C}_n}\otimes \rho_{E^n})
=\frac{1}{\alpha-1}\log \mathbb{E}_{\mathcal{C}_n}
\left[ Q_\alpha(\rho_{E^n}^{\mathcal{C}_n}\|\rho_{E^n})\right].
\end{equation*}
The reliability function of quantum soft covering for i.i.d. random codebooks is defined as
\begin{equation*}
E_{\rm sc}^{(\alpha)}(\rho_{XE},R):=\liminf_{n\to\infty} -\frac{1}{n} \log D_{\alpha} (\rho_{\mathcal{C}_nE^n}\|\rho_{\mathcal{C}_n}\otimes \rho_E^{\otimes n}).
\end{equation*}

\begin{theorem}\label{thm:ref-fun}
Let $\alpha\in[2,\infty)$, $R\geq0$ and $\rho_{XE}=\sum_{x\in\mathcal{ X}}P_{X}(x)|x\rangle\langle x|\otimes\rho_{E}^{x}$
be a C-Q state. When $R> I_\alpha(X:E)_{\rho_{XE}}$, we have
\begin{align*}
E_{\rm sc}^{(\alpha)}(\rho_{XE},R)= \min\{\gamma'(\alpha),\gamma(\alpha-1)\},
\end{align*}
where $\gamma(s):=s(R-I_{1+s}(X:E)_{\rho_{XE}})$ and $\gamma'(s):=R-I_{2}^{(s)}(X:E)_{\rho_{XE}}$.
\end{theorem}
\begin{remark}
We do not consider the case when $\rho_{XE}=\rho_{X}\otimes\rho_E$. In this case, the induced average output state is exactly the target state.
\end{remark}
\begin{remark}
At \(\alpha=2\), a direct collision calculation gives the exact identity
\begin{equation*}
\mathbb E_{\mathcal C}
Q_2\left(\rho_E^{\mathcal C}\middle\|\rho_E\right)
=
1+\frac{2^{I_2(X:E)_{\rho_{XE}}}-1}{M}.
\end{equation*}
Hence
\begin{equation*}
E_{\rm sc}^{(2)}(\rho_{XE},R)= R-I_2(X:E)_{\rho_{XE}}=\gamma(1)=\gamma'(2).
\end{equation*}
In the following, we establish the case $\alpha\in(2,\infty)$.
\end{remark}

\subsection{Reliability function for $\alpha\in(2,\infty)$}
In this subsection, we establish the reliability function of quantum soft covering for i.i.d. random codebooks under the sandwiched Rényi divergence of order $\alpha \in (2,\infty)$.

Throughout this subsection, we fix \(\alpha\in(2,\infty)\) and consider
the classical-quantum state
\begin{equation*}
\rho_{XE}
=
\sum_{x\in\mathcal X}
P_X(x)\,|x\rangle\langle x|\otimes\rho_E^x,
\end{equation*}
with marginal state
\begin{equation*}
\rho_E=\sum_{x\in\mathcal X}P_X(x)\rho_E^x.
\end{equation*}
We use the following notation throughout. All operators associated with
the quantum system \(E\) are understood to act on
\(\operatorname{supp}\rho_E\). Consequently, \(\rho_E\) is strictly
positive on the underlying Hilbert space. Define
\begin{equation*}
A_x
:=
\rho_E^{\frac{1-\alpha}{2\alpha}}
\rho_E^x
\rho_E^{\frac{1-\alpha}{2\alpha}},
\qquad
\overline A
:=
\rho_E^{1/\alpha},
\end{equation*}
and introduce the centered operators
\begin{equation*}
H_x:=A_x-\overline A.
\end{equation*}
By construction,
\begin{equation}
\mathbb E_X H_X
=
\sum_{x\in\mathcal X}P_X(x)H_x
=
0.
\label{eq:centered-Hx}
\end{equation}
Moreover,
\begin{align*}
&\sum_{x\in\mathcal X}
P_X(x)
\operatorname{Tr}
H_x\rho_E^{\frac{\alpha-2}{\alpha}}H_x
\nonumber\\
&\quad=
\sum_{x\in\mathcal X}
P_X(x)
\operatorname{Tr}
A_x\rho_E^{\frac{\alpha-2}{\alpha}}A_x
-
\operatorname{Tr}\rho_E
\nonumber\\
&\quad=
\mathbb E_X
\operatorname{Tr}
A_X^2\rho_E^{\frac{\alpha-2}{\alpha}}
-1
\nonumber\\
&\quad=
\exp\!\left\{
I_2^{(\alpha)}(X:E)_{\rho_{XE}}
\right\}
-1.
\end{align*}
We shall also use the following centered moment:
\begin{equation*}
V_\alpha(X:E)_{\rho_{XE}}
:=
\sum_{x\in\mathcal X}
P_X(x)\|H_x\|_\alpha^\alpha
=
\mathbb E_X\operatorname{Tr}|H_X|^\alpha
\ge 0.
\end{equation*}

Let
\(\mathcal C=\{X_1,\ldots,X_M\}\) be an i.i.d. random codebook whose
codewords \(X_1,\ldots,X_M\) are drawn according to \(P_X\). Define
\begin{equation*}
S_{\mathcal C}
:=
\frac{1}{M}\sum_{m=1}^M A_{X_m},
\qquad
\Delta_{\mathcal C}
:=
S_{\mathcal C}-\rho_E^{1/\alpha}
=
\frac{1}{M}\sum_{m=1}^M H_{X_m}.
\end{equation*}
It follows from \eqref{eq:centered-Hx} that
\begin{equation*}
\mathbb E_{\mathcal C}\Delta_{\mathcal C}=0.
\end{equation*}

\begin{proposition}\label{centered-moment-comparison}
The centered moment satisfies
\begin{align*}
\left(2^{1-\alpha}\exp\{(\alpha-1)I_\alpha(X:E)_{\rho_{XE}}\}-1\right)_{+}
&\leq V_\alpha(X:E)_{\rho_{XE}}\notag\\
&\leq 2^{\alpha}\exp\{(\alpha-1)I_\alpha(X:E)_{\rho_{XE}}\}.
\end{align*}
\end{proposition}

\begin{proof}
For the upper bound on $V_\alpha(X:E)_{\rho_{XE}}$, the Schatten norm triangle
inequality and the scalar inequality \(
(a+b)^\alpha
\leq
2^{\alpha-1}(a^\alpha+b^\alpha)
\) give
\begin{align*}
\Vert H_x \Vert_\alpha^\alpha=\left\|A_x-\bar{A}\right\|_\alpha^\alpha
\leq
\left(
\left\|A_x\right\|_\alpha
+
\left\|\bar{A}\right\|_\alpha
\right)^\alpha\leq
2^{\alpha-1}
\left(
\left\|A_x\right\|_\alpha^\alpha
+
\left\|\bar{A}\right\|_\alpha^\alpha
\right).
\end{align*}
Since
\(
\left\|\bar{A}\right\|_\alpha^\alpha
=
\operatorname{Tr}\bar{A}^{\alpha}
=1\), taking expectation over $x$ yields
\begin{align}\label{eq:centered-moment-upper}
V_\alpha(X:E)_{\rho_{XE}}
\leq&
2^{\alpha-1}(\exp\{(\alpha-1)I_\alpha(X:E)_{\rho_{XE}}\}+1)\nonumber\\
\leq&2^{\alpha}\exp\{(\alpha-1)I_\alpha(X:E)_{\rho_{XE}}\}.
\end{align}
In the last inequality we use the fact that $\exp\{(\alpha-1)I_\alpha(X:E)_{\rho_{XE}}\}\geq 1$.

For the lower bound, the triangle inequality gives
\begin{equation*}
\left\|A_x\right\|_\alpha
\leq
\left\|A_x-\bar{A}\right\|_\alpha
+
\left\|\bar{A}\right\|_\alpha.
\end{equation*}
Applying the same scalar inequality, we obtain
\begin{equation*}
\left\|A_x\right\|_\alpha^\alpha
\leq
2^{\alpha-1}
\left(
\left\|A_x-\bar{A}\right\|_\alpha^\alpha
+
\left\|\bar{A}\right\|_\alpha^\alpha
\right),
\end{equation*}
or equivalently,
\begin{equation*}
\left\|A_x-\bar{A}\right\|_\alpha^\alpha
\geq
2^{1-\alpha}\left\|A_x\right\|_\alpha^\alpha
-
\left\|\bar{A}\right\|_\alpha^\alpha.
\end{equation*}
Taking expectation on $x$ and using
$\left\|\bar{A}\right\|_\alpha^\alpha=1$, we obtain
\begin{equation*}
V_\alpha(X:E)_{\rho_{XE}}
\geq
2^{1-\alpha}\exp\{(\alpha-1)I_\alpha(X:E)_{\rho_{XE}}\}-1.
\end{equation*}
Since $V_\alpha(X:E)_{\rho_{XE}}\geq0$, this implies
\begin{equation}\label{lower-V}
V_\alpha(X:E)_{\rho_{XE}}
\geq
\left(2^{1-\alpha}\exp\{(\alpha-1)I_\alpha(X:E)_{\rho_{XE}}\}-1\right)_{+}.
\end{equation}
Combining \eqref{eq:centered-moment-upper} and \eqref{lower-V} completes the proof.
\end{proof}

\begin{proposition}\label{prop:two-sided-one-shot-detailed}
Let $\alpha>2$, $R\geq0$ and $\rho_{XE}$ be a C-Q state. Let $M=2^R$ and \(\mathcal C=\{X_1,\cdots,X_M\}\) be a random codebook whose codewords are sampled independently according to \(P_X\). There exists a constant \(C_\alpha<\infty\), depending only on
\(\alpha\), such that
\begin{align}
 Q_\alpha(\rho_{\mathcal CE}\Vert
          \rho_{\mathcal C}\otimes\rho_E)-1\leq C_\alpha\left(2^{-\gamma^\prime(\alpha)}+2^{-\gamma(\alpha-1)}
 \right)
 \label{eq:add-one-shot-upper}
\end{align}
and
\begin{align*}
 Q_\alpha(\rho_{\mathcal CE}\Vert
          \rho_{\mathcal C}\otimes\rho_E)\!-\!1\!\geq\!
 \max\Big\{
 (\alpha-1)(2^{-\gamma^\prime(\alpha)}-2^{-R}),
2^{2-2\alpha-\gamma(\alpha-1)}-2^{-(\alpha-1)(R+1)}
 \Big\},
\end{align*}
where $\gamma(s):=s(R-I_{1+s}(X:E)_{\rho_{XE}})$ and $\gamma'(s):=R-I_{2}^{(s)}(X:E)_{\rho_{XE}}$.
\end{proposition}

\begin{proof}
  Since $\E_{\mathcal C}\Delta_{\mathcal C}=0$ and
$\Tr\bar A^\alpha=1$,
\begin{align}
 Q_\alpha(\rho_{\mathcal CE}\Vert\rho_{\mathcal C}\otimes\rho_E)-1
 =\E_{\mathcal C} \Tr S_{\mathcal C}^{\alpha}-\Tr\bar A^\alpha=\E_{\mathcal C}
   \cB_\alpha(S_{\mathcal C},\bar A).
 \label{eq:add-Q-bregman}
\end{align}
Here, $\cB_\alpha$ is defined as in \eqref{eq:add-bregman-def}. We then divide the proof into upper bound and lower bound.

\noindent\textbf{\emph{Upper bound:}}
Lemma~\ref{lem:add-bregman} therefore implies that there is a
constant $K_\alpha<\infty$, depending only on $\alpha$, such that
\begin{align}
 Q_\alpha(\rho_{\mathcal CE}\Vert\rho_{\mathcal C}\otimes\rho_E)-1
 \leq K_\alpha\left[
 \E_{\mathcal C}\Tr\Delta_{\mathcal C}
                 \bar A^{\alpha-2}\Delta_{\mathcal C}
 +\E_{\mathcal C}\Tr|\Delta_{\mathcal C}|^\alpha
 \right].
 \label{eq:add-upper-start}
\end{align}
Since \(X_1,\ldots,X_M\) are independent and
\(\mathbb{E}_{X} H_X=0\), we have
\begin{align*}
\mathbb{E}_{X_i,X_j}
\Tr
  H_{X_i}
  \bar{A}^{\alpha-2}
  H_{X_j}
=
\Tr \mathbb{E}_{X_i}
  H_{X_i}
  \overline{A}^{\alpha-2}
  \mathbb{E}_{X_j}H_{X_j}
=0,\quad i\neq j.
\end{align*}
Hence
\begin{align*}
\mathbb{E}_{\mathcal{C}}
\Tr\Delta_{\mathcal{C}}\bar{A}^{\alpha-2}\Delta_{\mathcal{C}}
=&\frac{1}{M^2}\sum_{i,j=1}^M\mathbb{E}_{\mathcal{C}}\Tr H_{X_i}\bar{A}^{\alpha-2}H_{X_j} \nonumber\\
=&\frac{1}{M^2}\sum_{i=1}^M\mathbb{E}_{X_i}\operatorname{Tr}H_{X_i}\bar{A}^{\alpha-2}H_{X_i}
\end{align*}
Since $X_1,\cdots,X_M$ are identically distributed, we get
\begin{align*}
\frac{1}{M^2}
\sum_{i=1}^M
\mathbb{E}_{X_i}
\operatorname{Tr}
  H_{X_i}
  \bar{A}^{\alpha-2}
  H_{X_i}
=
\frac{1}{M}
\mathbb{E}_X
\operatorname{Tr}
  H_X
  \bar{A}^{\alpha-2}
  H_X.
\end{align*}
Hence we obtain
\begin{align}\label{eq:add-kms-average}
 \E_{\mathcal C}\Tr\Delta_{\mathcal C}
 \bar A^{\alpha-2}\Delta_{\mathcal C}=&\frac{1}{M}
\mathbb{E}_X
\operatorname{Tr}
  H_X
  \bar{A}^{\alpha-2}
  H_X \nonumber\\
  =&\frac{1}{M}(\exp{I_2^{(\alpha)}(X:E)_{\rho_{XE}}}-1)\leq 2^{-\gamma^\prime(\alpha)}.
\end{align}

For the other item in~\eqref{eq:add-upper-start}, writing
$$\Delta_\mathcal{C}=\sum_{i=1}^{M}\frac{H_{X_i}}{M}$$
and applying the matrix Rosenthal inequality
Lemma~\ref{lem:matrix-rosenthal} to $\frac{1}{M}H_{X_m}$ gives that there exists a finite constant \(T_\alpha\), depending only on \(\alpha\), such that
\begin{equation}\label{eq:2}
 \E_{\mathcal C}\Tr|\Delta_{\mathcal C}|^\alpha
 \leq T_\alpha\left[\frac{V_\alpha(X:E)_{\rho_{XE}}}{M^{\alpha-1}}+
 \frac{\Tr(\E_X H_X^2)^{\alpha/2}}{M^{\alpha/2}}
 \right].
\end{equation}
Since $\E_X H_X^2=\E_X A_X^2-\bar A^2$, we get
$$0\leq \E_X H_X^2\leq \E_X A_X^2.$$
Eigenvalue monotonicity yields
\begin{equation*}
 \Tr (\E_X H_X^2)^{\alpha/2}\leq\Tr (\E_X A_X^2)^{\alpha/2}.
\end{equation*}
Lemma~\ref{lem:add-positive-bridge} then gives
\begin{equation*}
\begin{aligned}
 \Tr (\E_X A_X^2)^{\alpha/2}&\leq\sqrt{
 \left(\Tr \E_X A_X^2\bar A^{\alpha-2}\right)
 \left(\E_X\Tr A_X^\alpha\right)}\\
 &=\sqrt{
 \exp\{I_2^{(\alpha)}(X:E)_{\rho_{XE}}\} \exp\{(\alpha-1)I_{\alpha}(X:E)_{\rho_{XE}}\}}.
\end{aligned}
\end{equation*}
This further yields
\begin{equation}\label{est-2}
\begin{aligned}
\frac{\Tr (\E_X H_X^2)^{\alpha/2}}{M^{\alpha/2}}&\leq \frac{\sqrt{
 \exp\{I_2^{(\alpha)}(X:E)_{\rho_{XE}}\} \exp{\{(\alpha-1)I_{\alpha}(X:E)_{\rho_{XE}}}\}}}{M^{\alpha/2}}\\
 &= \sqrt{2^{-\gamma^\prime(\alpha)}\cdot 2^{-\gamma(\alpha-1)}}\\
 &\leq 2^{-\gamma^\prime(\alpha)}+ 2^{-\gamma(\alpha-1)}
\end{aligned}
\end{equation}
Finally, by Proposition~\ref{centered-moment-comparison}, we get
\begin{align}\label{est-3}
\frac{V_\alpha(X:E)_{\rho_{XE}}}{M^{\alpha-1}}\leq \frac{2^\alpha\exp{\{(\alpha-1)I_{\alpha}(X:E)_{\rho_{XE}}\}}}{M^{\alpha-1}}=2^\alpha \cdot 2^{-\gamma(\alpha-1)}.
\end{align}
The combination of \eqref{eq:2}, \eqref{est-2} and \eqref{est-3} yields
\begin{align}\label{est-4}
 \E_{\mathcal C}\Tr|\Delta_{\mathcal C}|^\alpha
\leq& T_\alpha\left[2^\alpha \cdot 2^{-\gamma(\alpha-1)}+
2^{-\gamma^\prime(\alpha)}+ 2^{-\gamma(\alpha-1)}
\right] \nb\\
\leq&(2^\alpha+1)T_\alpha\left[2^{-\gamma(\alpha-1)}+
2^{-\gamma^\prime(\alpha)}\right].
\end{align}
Substituting of \eqref{eq:add-kms-average},
and \eqref{est-4} into
\eqref{eq:add-upper-start} proves 
\begin{align}
Q_\alpha(\rho_{\mathcal CE}\Vert\rho_{\mathcal C}\otimes\rho_E)-1
\leq& K_\alpha\left[
2^{-\gamma^\prime(\alpha)}
+(2^\alpha+1)T_\alpha\left(2^{-\gamma(\alpha-1)}+
2^{-\gamma^\prime(\alpha)}\right)\right] \nb\\
\leq&K_\alpha\left((2^\alpha+1)T_\alpha+1\right)\left[2^{-\gamma(\alpha-1)}+
2^{-\gamma^\prime(\alpha)}\right].
\end{align}
Setting $
C_\alpha:=K_\alpha\bigl((2^\alpha+1)T_\alpha+1\bigr),$
which depends only on \(\alpha\), yields
\eqref{eq:add-one-shot-upper}.

\noindent\textbf{\emph{Lower bound:}}
For the first lower bound, applying the lower bound of
Lemma~\ref{lem:add-bregman} in \eqref{eq:add-Q-bregman} with
\eqref{eq:add-kms-average}, we get
\begin{equation*}
 Q_\alpha(\rho_{\mathcal CE}\Vert\rho_{\mathcal C}\otimes\rho_E)-1\geq\frac{\alpha-1}{M}(\exp\{I^{(\alpha)}_2(X:E)_{\rho_{XE}}\}-1)=(\alpha-1)(2^{-\gamma^\prime(\alpha)}-2^{-R}).
\end{equation*}

To prove another lower bound, for \(k=0,\ldots,M\), define
\[
T_k
:=
\bar A+\frac1M\sum_{m=1}^kH_{X_m}
=
\frac1M\sum_{m=1}^kA_{X_m}
+
\left(1-\frac{k}{M}\right)\bar A.
\]
Then
\[
T_0=\bar A,\qquad
T_M=S_{\mathcal C},\qquad
T_k\geq0.
\]
Let
\(
\mathcal F_k=\sigma(X_1,\ldots,X_k)
\). Since
\[
\mathbb E[T_{k}\mid\mathcal F_{k-1}]=E\left[T_{k-1}+\frac1M H_{X_k}\Big\vert\mathcal F_{k-1}\right]=T_{k-1},
\]
Thus, $\{T_k\}_{k=0}^{M}$ is a martingale with respect to $\{\mathcal{F}_k\}_{k=0}^{M}$. Hence
$$ \mathbb E[T_{k-1}^{\alpha-1}(T_{k}-T_{k-1})\mid\mathcal F_{k-1}]=T_{k-1}^{\alpha-1} \mathbb E\left[\frac{1}{M} H_{X_k} \Big\vert \mathcal F_{k-1}\right]=0.$$
Since
$$ \cB_\alpha(T_k,T_{k-1})=\Tr T_k^{\alpha}-\Tr T_{k-1}^\alpha-\alpha \Tr T_{k-1}^{\alpha-1}(T_{k}-T_{k-1}),$$
we get that
\begin{align*}
&\mathbb E\left[
\operatorname{Tr}T_k^\alpha
-\operatorname{Tr}T_{k-1}^\alpha
\mid\mathcal F_{k-1}
\right]=
\mathbb E\left[
\mathfrak B_\alpha(T_k,T_{k-1})
\mid\mathcal F_{k-1}
\right].
\end{align*}
The lower bound estimate in Proposition~\ref{b-est} gives
\begin{align*}
&\mathbb E\left[
\operatorname{Tr}T_k^\alpha
-\operatorname{Tr}T_{k-1}^\alpha
\mid\mathcal F_{k-1}
\right]\geq
2^{1-\alpha}
\mathbb E\left[
\|T_k-T_{k-1}\|_\alpha^\alpha
\mid\mathcal F_{k-1}
\right]=
\frac{2^{1-\alpha}}{M^\alpha}
V_\alpha(X:E)_{\rho_{XE}}.
\end{align*}
Summing over \(k=1,\ldots,M\), we obtain
\begin{align*}
Q_\alpha(\rho_{\mathcal{C}E}\Vert \rho_\mathcal{C}\otimes\rho_E)-1=&
\mathbb E\operatorname{Tr}T_M^\alpha
-\operatorname{Tr}T_0^\alpha \nonumber\\
\geq&
\frac{2^{1-\alpha}}{M^{\alpha-1}}
V_\alpha(X:E)_{\rho_{XE}}\nonumber\\
\geq&2^{2-2\alpha-\gamma(\alpha-1)}-2^{-(\alpha-1)(R+1)},
\end{align*}
where the last inequality follows from Proposition~\ref{centered-moment-comparison}.
Combining the two lower bounds proves the desired result.
\end{proof}

From the above one-shot bound, we establish Theorem~\ref{thm:ref-fun}.
\begin{proof}[Proof of Theorem~\ref{thm:ref-fun}:]
For any $n\in\mathbb{N}$,
let $\mathcal{C}_{n}:=\{ X^{n}(1),...,X^{n}(2^{nR})\} $ be
an i.i.d. random codebook,
where each $X^{n}(m)$ is independently drawn from $P_{X}^{\otimes n}$.
Applying Proposition~\ref{prop:two-sided-one-shot-detailed} with the substitutions $\mathcal{C} \leftarrow \mathcal{C}_n$ and $\rho_{XE}\leftarrow \rho_{XE}^{\otimes n}$, and making use of the additivity of sandwiched R{\'e}nyi mutual information and the mixed-order order-two R{\'e}nyi mutual information, we get
\begin{align*}
Q_\alpha(\rho_{\mathcal{C}_{n}E^{n}}\|\rho_{\mathcal{C}_{n}}\otimes\rho_{E}^{\otimes n})-1 \leq C_\alpha(2^{-n\gamma^\prime(\alpha)}+2^{-n\gamma(\alpha-1)}).
\end{align*}
and
\begin{align*}
&Q_\alpha(\rho_{\mathcal{C}_{n}E^{n}}\|\rho_{\mathcal{C}_{n}}\otimes\rho_{E}^{\otimes n})-1 \nb\\
\ge&  \max\left\{
(\alpha-1)(2^{-n\gamma^\prime(\alpha)}-2^{-nR}),
2^{2-2\alpha-n\gamma(\alpha-1)}-2^{-(\alpha-1)(nR+1)}
\right\}.
\end{align*}
In both directions we use Proposition~\ref{centered-moment-comparison} and
\(\log(1+f(n))\sim f(n)/\ln 2\) as \(f(n)\searrow0\). If the rate $R$ satisfies $R>  I_{\alpha}(X:E)_{\rho_{XE}}$, we have
\begin{align*}
&\liminf_{n\to\infty}-\frac{1}{n}\log D_{\alpha}(\rho_{\mathcal{C}_{n}E^{n}}\|\rho_{\mathcal{C}_{n}}\otimes\rho_{E}^{\otimes n})\nb\\
\geq& \min\{\gamma'(\alpha),\gamma(\alpha-1)\}
\end{align*}
and
\begin{align}
\limsup_{n\to\infty}-\frac{1}{n}\log D_\alpha(\rho_{\mathcal{C}_{n}E^{n}}\|\rho_{\mathcal{C}_{n}}\otimes\rho_{E}^{\otimes n})
\leq  \min\{\gamma'(\alpha),\gamma(\alpha-1)\},\label{eq3}
\end{align}
where Equation~\eqref{eq3} follows from the fact that 
\begin{align*}
\gamma(\alpha-1)=&(\alpha-1)(R-I_\alpha(X:E)_{\rho_{XE}})<(\alpha-1)R, \nonumber\\
\gamma'(\alpha)=&R-I_{2}^{(\alpha)}(X:E)_{\rho_{XE}}<R. \nonumber
\end{align*}
By the definition of reliability function of quantum soft covering for i.i.d. random code, the desired result follows.
\end{proof}

\section{Reliability Function of Quantum Privacy Amplification}\label{sec:pa}
In this section, we derive the reliability function of quantum privacy amplification via the random binning function under the sandwiched R{\'e}nyi divergence for orders $\alpha \in [2,\infty)$.
\subsection{Problem Statement and Main Result}
Let $\rho_{XE}=\sum_{x\in\mathcal{X}} P_X(x)\, |x\rangle\langle x| \otimes \rho_E^x$
be a C-Q state.
The goal of quantum privacy amplification is to extract from $X$ a key that is approximately uniform and independent of the quantum system $E$. A random function
\[
F:\mathcal X\to\mathcal Z=\{1,2,\cdots,M\}
\]
is called a random binning function if the random variables
$\{F(x)\}_{x\in\mathcal X}$ are independent and uniformly distributed
over $\mathcal Z$. Equivalently, for any $x\in\mathcal{X}$ and $z\in\mathcal{Z}$,
\begin{equation*}
{\rm Pr}_{F}\{F(x)=z\} = \frac{1}{M},
\end{equation*}
and for any finite subset $\{x_1,x_2,\cdots,x_k\}\subset\mathcal{X}$, the random variables $F(x_1),\cdots,F(x_k)$ are mutually independent.
The C-Q state induced by the random binning function $F$ is
\begin{equation*}
\mathcal{R}_{F}(\rho_{XE})
=\sum_{x\in\mathcal{X}} P_X(x)
|F(x)\rangle\langle F(x)| \otimes \rho_{E}^{x}.
\end{equation*}
Note that $\mathcal{R}_{F}(\rho_{XE})$ is a random state due to the randomness of $F$.
The ideal target state is $
\frac{\mathbbm{1}_{\mathcal{Z}}}{|\mathcal{Z}|}\otimes \rho_E$ with $
\rho_E=\sum_{x\in\mathcal{X}}P_X(x)\rho_E^x.$

To quantify the performance of quantum privacy amplification, we measure the discrepancy via the sandwiched Rényi divergence of order $\alpha$. That is
\begin{equation*}
\frac{1}{\alpha-1}\log \mathbbm{E}_F Q_\alpha\left(\mathcal{R}_{F}(\rho_{XE})\Big\|\frac{\mathbbm{1}_{\mathcal{Z}}}{|\mathcal{Z}|}\otimes \rho_E\right).
\end{equation*}

We consider the $n$-shot extension of quantum privacy amplification. Let $
F_n:\mathcal{X}^n \to \mathcal{Z}_n=\{1,2,\cdots,2^{nR}\}$ be the random binning function, where $R\ge 0$ is the key rate. The state induced by $F_n$ is
\begin{equation*}
\mathcal{R}_{F_n}(\rho_{XE}^{\otimes n})
=\sum_{x^n\in\mathcal{X}^n} P_X^{\otimes n}(x^n)\,
|F_n(x^n)\rangle\langle F_n(x^n)| \otimes \rho_{E^n}^{x^n},
\end{equation*}
where $\rho_{E^n}^{x^n}=\bigotimes_{i=1}^n \rho_E^{x_i}$.
The performance in the $n$-shot setting is characterized by
\begin{equation*}
\frac{1}{\alpha-1}\log \mathbbm{E}_{F_n} Q_\alpha\left(\mathcal{R}_{F_n}(\rho_{XE}^{\otimes n})\Big\|\frac{\mathbbm{1}_{\mathcal{Z}_n}}{|\mathcal{Z}_n|}\otimes \rho_E^{\otimes n}\right).
\end{equation*}
The reliability function of quantum privacy amplification is defined as
\begin{align*}
E_{\rm pa}^{(\alpha)}(\rho_{XE},R)
:=&\liminf_{n\to\infty}
-\frac{1}{n}\log\left\{
\frac{1}{\alpha-1}\log \mathbb{E}_{F_n} Q_\alpha\left(\mathcal{R}_{F_n}(\rho_{XE}^{\otimes n})\Big\|\frac{\mathbbm{1}_{\mathcal{Z}_n}}{|\mathcal{Z}_n|}\otimes \rho_E^{\otimes n}\right)\right\}.
\end{align*}

\begin{theorem}\label{thm:10}
Let $\alpha\geq2$ and $\rho_{XE}=\sum_{x\in\mathcal{X}} P_X(x)\, |x\rangle\langle x| \otimes \rho_E^x$
be a C-Q state. When $0<R<H_\alpha(X|E)_{\rho_{XE}}$, we have
\begin{equation*}
E_{\rm pa}^{(\alpha)}(\rho_{XE},R)=\min\{\eta'(\alpha),\eta(\alpha-1)\},
\end{equation*}
where $\eta(t):=t(H_{1+t}\left(X|E\right)_{\rho_{XE}}-R)$ and $\eta'(t):=H_{2}^{(t)}(X|E)_{\rho_{XE}}-R$.
\end{theorem}
\begin{remark}
At \(\alpha=2\), a direct collision calculation gives the exact
identity
\begin{equation*}
\mathbb E_FQ_2\left(
\mathcal R_F(\rho_{XE})
\middle\|
\frac{\mathbbm{1}_{\mathcal{Z}}}{|\mathcal{Z}|}\otimes\rho_E
\right)
=
1+(M-1)2^{-H_2(X|E)_\rho}.
\end{equation*}
Hence
\begin{equation*}
E_\mathrm{pa}^{(2)}(\rho_{XE},R)
=H_2(X|E)_\rho-R=\eta(1)=\eta'(2).
\end{equation*}
\end{remark}
\subsection{Reliability Function for $\alpha\in(2,\infty)$}
We next prove the reliability function for $\alpha\in(2,\infty)$. To this end, we first establish the following proposition.

\begin{proposition}\label{prop:Bernoulli}
Let \(\alpha\in(2,\infty)\), \(A_1,\cdots,A_N\ge0\), and
\(\overline A=\sum_iA_i\).  Let
\(\xi_1,\cdots,\xi_N\) be independent Bernoulli random variables with
\(\Pr\{\xi_i=1\}=q\in(0,1)\).  Define
\begin{equation*}
S:=\sum_{i=1}^N\xi_iA_i,\qquad
B:=\mathbb E S=q\overline A, \qquad\Xi_\alpha
:=
\sum_{i=1}^N\operatorname{Tr} A_i^2\overline A^{\alpha-2},
\qquad
\Theta_\alpha
:=
\sum_{i=1}^N\operatorname{Tr} A_i^\alpha.
\end{equation*}
There exists a constant \(C_\alpha<\infty\), depending only on
\(\alpha\), such that
\begin{align}
&\mathbb E\operatorname{Tr} S^\alpha-\operatorname{Tr} B^\alpha \nonumber\\
\le&
C_\alpha\left[
q^{\alpha-1}(1-q)\Xi_\alpha
+
[q(1-q)]^{\alpha/2}
\sqrt{\Xi_\alpha\Theta_\alpha}
+
(q(1-q)^\alpha+(1-q)q^\alpha)\Theta_\alpha
\right]
\label{eq:Bernoulli-upper}
\end{align}
and
\begin{align}
\mathbb E\operatorname{Tr} S^\alpha-\operatorname{Tr} B^\alpha
\ge
\max\left\{
(\alpha-1)q^{\alpha-1}(1-q)\Xi_\alpha,\,
2^{-\alpha}(q(1-q)^\alpha+(1-q)q^\alpha)\Theta_\alpha
\right\}.
\label{eq:Bernoulli-lower}
\end{align}
\end{proposition}

\begin{proof}
Define
\begin{equation*}
Y_i:=(\xi_i-q)A_i,\qquad
\Delta:=S-B=\sum_{i=1}^NY_i.
\end{equation*}
The operators \(\{Y_i\}_{i=1}^N\) are independent, self-adjoint, and centered.
Since \(\mathbb E\Delta=0\), we have
\begin{align}
\mathbb E\operatorname{Tr} S^\alpha-\operatorname{Tr} B^\alpha
&=
\mathbb E\mathfrak B_\alpha(S,B).
\label{eq:gap-Bregman}
\end{align}
Here $\cB_\alpha$ follows the definition in \eqref{eq:add-bregman-def}.

\noindent\textbf{\emph{Upper bound:}}
By Lemma~\ref{lem:add-bregman}, there exists a constant $C^\prime_\alpha$ depending only on $\alpha$ such that
\begin{align}
\mathbb E\operatorname{Tr} S^\alpha-\operatorname{Tr} B^\alpha
\le
C_\alpha'\left[
\mathbb E\operatorname{Tr}\Delta B^{\alpha-2}\Delta
+
\mathbb E\operatorname{Tr}|\Delta|^\alpha
\right].
\label{eq:gap-upper-Bregman}
\end{align}
Independence and centering make all cross terms vanish.  Therefore
\begin{align}
\mathbb E\operatorname{Tr}\Delta B^{\alpha-2}\Delta
&=
\sum_{i=1}^N
\mathbb E\operatorname{Tr} Y_iB^{\alpha-2}Y_i
\notag\\
&=
q(1-q)q^{\alpha-2}
\sum_{i=1}^N
\operatorname{Tr} A_i^2\overline A^{\alpha-2}
\notag\\
&=
q^{\alpha-1}(1-q)\Xi_\alpha.
\label{eq:quadratic-exact}
\end{align}
Applying the non-commutative Rosenthal inequality (cf. Lemma~\ref{lem:matrix-rosenthal}), there exists a constant $C_\alpha''$ depending only on $\alpha$ such that
\begin{align}
\mathbb E\operatorname{Tr}|\Delta|^\alpha
\le
C_\alpha''\left[
\operatorname{Tr}\left(\sum_i\mathbb E Y_i^2\right)^{\alpha/2}
+
\sum_i\mathbb E\operatorname{Tr}|Y_i|^\alpha
\right].
\label{eq:Rosenthal-Delta}
\end{align}
We first estimate the first term in brackets of~\eqref{eq:Rosenthal-Delta}.
Lemma~\ref{lem:add-positive-bridge} yields
\begin{align*}
\operatorname{Tr} \left(\sum_iA_i^2\right)^{\alpha/2}
\leq& \sqrt{\operatorname{Tr} \sum_iA_i^2\overline A^{\alpha-2} \left(\operatorname{Tr}\sum_iA_i^\alpha \right)} \nonumber\\
=&\sqrt{\Xi_\alpha\Theta_\alpha}.
\end{align*}
Consequently,
\begin{equation*}
\operatorname{Tr}\left(\sum_i\mathbb E Y_i^2\right)^{\alpha/2} =  \operatorname{Tr}\left(q(1-q)\sum_iA_i^2\right)^{\alpha/2}
\le
[q(1-q)]^{\alpha/2}
\sqrt{\Xi_\alpha\Theta_\alpha}.
\end{equation*}
The last term in brackets of~\eqref{eq:Rosenthal-Delta} satisfies
\begin{align}
\sum_i\mathbb E\operatorname{Tr}|Y_i|^\alpha
&=
\left[
q(1-q)^\alpha+(1-q)q^\alpha
\right]
\sum_i\operatorname{Tr} A_i^\alpha
\notag\\
&=
( q(1-q)^\alpha+(1-q)q^\alpha)\Theta_\alpha.
\label{eq:jump-exact}
\end{align}
Therefore, we obtain
\begin{equation}
\mathbb E\operatorname{Tr}|\Delta|^\alpha
\le
C_\alpha''\left[
[q(1-q)]^{\alpha/2}
\sqrt{\Xi_\alpha\Theta_\alpha}
+
( q(1-q)^\alpha+(1-q)q^\alpha)\Theta_\alpha
\right]. \label{eq:square-function-final}
\end{equation}
Substitution of
\eqref{eq:quadratic-exact} and
\eqref{eq:square-function-final} into
\eqref{eq:gap-upper-Bregman} establishes that there exists a constant $C_\alpha$ depending only on $\alpha$ such that
\eqref{eq:Bernoulli-upper} is valid.

\noindent\textbf{\emph{Lower bound:}}
We prove the two lower bounds in \eqref{eq:Bernoulli-lower}
separately.

First, the lower estimate in Lemma~\ref{lem:add-bregman} gives
\begin{align*}
\mathfrak B_\alpha(B+\Delta,B)
\ge
(\alpha-1)\operatorname{Tr}
\Delta B^{\alpha-2}\Delta.
\end{align*}
Taking expectations and using \eqref{eq:gap-Bregman} and
\eqref{eq:quadratic-exact}, we obtain
\begin{align}
\mathbb E\operatorname{Tr}S^\alpha-\operatorname{Tr}B^\alpha
&=
\mathbb E\mathfrak B_\alpha(B+\Delta,B)
\notag\\
&\ge
(\alpha-1)
\mathbb E\operatorname{Tr}\Delta B^{\alpha-2}\Delta
\notag\\
&=
(\alpha-1)q^{\alpha-1}(1-q)\Xi_\alpha.
\label{eq:lower-quadratic}
\end{align}
This proves the first lower bound.

We next prove another lower bound involving \(\Theta_\alpha\).
Let
\[
S'=B+\Delta'
\]
be an independent copy of \(S=B+\Delta\). In particular,
\(S\) and \(S'\) are positive, \(\Delta\) and \(\Delta'\) are
independent and identically distributed, and
\[
\mathbb E\Delta=\mathbb E\Delta'=0.
\]
For \(p\ge2\), the upper Clarkson--McCarthy inequality (cf. Lemma~\ref{lem:CM}) states that
\begin{equation}
\left\|\frac{U+V}{2}\right\|_p^p
+
\left\|\frac{U-V}{2}\right\|_p^p
\le
\frac{1}{2}
\left(
\|U\|_p^p+\|V\|_p^p
\right).
\label{eq:Clarkson-upper}
\end{equation}
Applying \eqref{eq:Clarkson-upper} with \(U=S\), \(V=S'\), and
\(p=\alpha\), and observing that
\[
\frac{S+S'}{2}
=
B+\frac{\Delta+\Delta'}{2},
\qquad
\frac{S-S'}{2}
=
\frac{\Delta-\Delta'}{2},
\]
we obtain
\begin{align*}
\frac{1}{2}
\left(
\|S\|_\alpha^\alpha+\|S'\|_\alpha^\alpha
\right)
\ge
\left\|
B+\frac{\Delta+\Delta'}{2}
\right\|_\alpha^\alpha
+
2^{-\alpha}
\|\Delta-\Delta'\|_\alpha^\alpha.
\end{align*}
Since \(S\) and \(S'\) are identically distributed, taking
expectations of both sides gives
\begin{align}
\mathbb E\|S\|_\alpha^\alpha
\ge
\mathbb E
\left\|
B+\frac{\Delta+\Delta'}{2}
\right\|_\alpha^\alpha
+
2^{-\alpha}
\mathbb E\|\Delta-\Delta'\|_\alpha^\alpha.
\label{eq:Clarkson-expectation}
\end{align}
The map \(X\mapsto\|X\|_\alpha^\alpha\) is convex. Hence, by
Jensen's inequality,
\begin{align}
\mathbb E
\left\|
B+\frac{\Delta+\Delta'}{2}
\right\|_\alpha^\alpha
&\ge
\left\|
\mathbb E\left[
B+\frac{\Delta+\Delta'}{2}
\right]
\right\|_\alpha^\alpha
\notag\\
&=
\|B\|_\alpha^\alpha
=
\operatorname{Tr}B^\alpha.
\label{eq:Jensen-midpoint}
\end{align}
Moreover, since \(S\ge0\),
\[
\|S\|_\alpha^\alpha=\operatorname{Tr}S^\alpha.
\]
Therefore, \eqref{eq:Clarkson-expectation} and
\eqref{eq:Jensen-midpoint} imply the symmetrization estimate
\begin{align}
\mathbb E\operatorname{Tr}S^\alpha-\operatorname{Tr}B^\alpha
\ge
2^{-\alpha}
\mathbb E\|\Delta-\Delta'\|_\alpha^\alpha.
\label{eq:Clarkson-gap}
\end{align}
It remains to find a lower bound for the moment on the right-hand
side of \eqref{eq:Clarkson-gap}. Let \(Y_1',\cdots,Y_N'\) be
independent copies of \(Y_1,\cdots,Y_N\), so that
\[
\Delta'=\sum_{i=1}^N Y_i',
\]
and define
\begin{equation*}
D_i:=Y_i-Y_i'.
\end{equation*}
Then
\[
\Delta-\Delta'=\sum_{i=1}^N D_i.
\]
The random operators \(D_1,\cdots,D_N\) are independent.
Furthermore, each \(D_i\) is symmetric in distribution because
\(Y_i\) and \(Y_i'\) are identically distributed:
\[
D_i=Y_i-Y_i'
\stackrel{\mathrm d}{=}
Y_i'-Y_i=-D_i.
\]
Let \(\varepsilon_1,\cdots,\varepsilon_N\) be independent
Rademacher random variables, independent of all the \(D_i\)'s.
The symmetry and independence of the \(D_i\)'s imply
\[
(D_1,\cdots,D_N)
\stackrel{\mathrm d}{=}
(\varepsilon_1D_1,\cdots,\varepsilon_ND_N).
\]
Consequently,
\begin{align}
\mathbb E\|\Delta-\Delta'\|_\alpha^\alpha=
\mathbb E_D
\left\|
\sum_{i=1}^N D_i
\right\|_\alpha^\alpha=
\mathbb E_D\mathbb E_\varepsilon
\left\|
\sum_{i=1}^N\varepsilon_iD_i
\right\|_\alpha^\alpha.
\label{eq:sign-insertion}
\end{align}
For \(p\ge2\), the lower Clarkson--McCarthy inequality (cf. Lemma~\ref{lem:CM}) states that
\begin{equation}
\|U+V\|_p^p+\|U-V\|_p^p
\ge
2\left(\|U\|_p^p+\|V\|_p^p\right).
\label{eq:Clarkson-lower}
\end{equation}
For completeness, we now explain how this inequality gives
\begin{equation}
\mathbb E_\varepsilon
\left\|
\sum_{i=1}^N\varepsilon_iD_i
\right\|_\alpha^\alpha
\ge
\sum_{i=1}^N\|D_i\|_\alpha^\alpha.
\label{eq:Clarkson-induction}
\end{equation}
Fix \(D_1,\cdots,D_N\), and set
\[
R_{k-1}:=\sum_{i=1}^{k-1}\varepsilon_iD_i.
\]
Conditioning on \(\varepsilon_1,\cdots,\varepsilon_{k-1}\) and
averaging only over \(\varepsilon_k\), we have
\begin{align}
\mathbb E_{\varepsilon_k}
\|R_{k-1}+\varepsilon_kD_k\|_\alpha^\alpha
&=
\frac{1}{2}
\left(
\|R_{k-1}+D_k\|_\alpha^\alpha
+
\|R_{k-1}-D_k\|_\alpha^\alpha
\right)
\notag\\
&\ge
\|R_{k-1}\|_\alpha^\alpha+\|D_k\|_\alpha^\alpha,
\label{eq:Clarkson-one-step}
\end{align}
where the last inequality follows from
\eqref{eq:Clarkson-lower}. Taking expectations over
\(\varepsilon_1,\cdots,\varepsilon_{k-1}\) and iterating
\eqref{eq:Clarkson-one-step} for \(k=2,\cdots,N\) proves
\eqref{eq:Clarkson-induction}.
Combining \eqref{eq:sign-insertion} and
\eqref{eq:Clarkson-induction}, we obtain
\begin{align*}
\mathbb E\|\Delta-\Delta'\|_\alpha^\alpha
\ge
\sum_{i=1}^N\mathbb E\|D_i\|_\alpha^\alpha.
\end{align*}
For each fixed value of \(Y_i\), the convexity of
\(X\mapsto\|X\|_\alpha^\alpha\) and conditional Jensen's
inequality give
\begin{align*}
\mathbb E_{Y_i'}
\|Y_i-Y_i'\|_\alpha^\alpha\ge
\left\|
Y_i-\mathbb E Y_i'
\right\|_\alpha^\alpha
=
\|Y_i\|_\alpha^\alpha,
\end{align*}
because \(\mathbb EY_i'=0\). Taking the expectation over \(Y_i\)
and summing over \(i\) yields
\begin{align*}
\sum_{i=1}^N\mathbb E\|D_i\|_\alpha^\alpha\ge
\sum_{i=1}^N\mathbb E\|Y_i\|_\alpha^\alpha
=
\left[
q(1-q)^\alpha+(1-q)q^\alpha
\right]\Theta_\alpha,
\end{align*}
where the last equality follows from \eqref{eq:jump-exact}.
Thus,
\begin{align}
\mathbb E\|\Delta-\Delta'\|_\alpha^\alpha
\ge
\left[
q(1-q)^\alpha+(1-q)q^\alpha
\right]\Theta_\alpha.
\label{eq:symmetrized-final}
\end{align}
Substituting \eqref{eq:symmetrized-final} into
\eqref{eq:Clarkson-gap}, we conclude that
\begin{align}
\mathbb E\operatorname{Tr}S^\alpha-\operatorname{Tr}B^\alpha
\ge
2^{-\alpha}
\left[
q(1-q)^\alpha+(1-q)q^\alpha
\right]\Theta_\alpha.
\label{eq:lower-jump}
\end{align}
Combining \eqref{eq:lower-quadratic} and \eqref{eq:lower-jump}
proves
\begin{align*}
\mathbb E\operatorname{Tr}S^\alpha-\operatorname{Tr}B^\alpha
\ge
\max\left\{
(\alpha-1)q^{\alpha-1}(1-q)\Xi_\alpha,\,
2^{-\alpha}
\left[
q(1-q)^\alpha+(1-q)q^\alpha
\right]\Theta_\alpha
\right\},
\end{align*}
as claimed.
\end{proof}
\begin{remark}
In the preceding proof, we established the lower bound~\eqref{eq:lower-jump} using Rademacher randomization. Alternatively, the same bound can be derived by adapting the martingale argument used in Proposition~\ref{prop:two-sided-one-shot-detailed}. Conversely, the Rademacher argument employed here can also be adapted to prove that proposition. We retain both proofs to highlight these two complementary techniques.

\end{remark}

Throughout the remainder of this subsection, for the classical-quantum
state
\[
\rho_{XE}
:=
\sum_{x\in\mathcal{X}}
P_X(x)\,|x\rangle\langle x|\otimes\rho_E^x,
\]
we use the following notation
\begin{equation}
A_x
:=
P_X(x)
\rho_E^{\frac{1-\alpha}{2\alpha}}
\rho_E^x
\rho_E^{\frac{1-\alpha}{2\alpha}},
\qquad
\overline A
:=
\sum_x A_x
=
\rho_E^{1/\alpha}.
\label{eq:Ax-PA}
\end{equation}
With this notation, we have
\begin{align}
\sum_x \operatorname{Tr} A_x^\alpha
&=
2^{-(\alpha-1)H_\alpha(X|E)_\rho},
\label{eq:Theta-entropy}\\
\sum_x \operatorname{Tr} A_x^2\overline A^{\alpha-2}
&=
2^{-H_2^{(\alpha)}(X|E)_\rho}.
\label{eq:Xi-entropy}
\end{align}

\begin{theorem}\label{thm:one-shot}
Let \(\alpha\in(2,\infty)\), $M=2^R\geq2$ and \(F:\mathcal{X}\to\mathcal{Z}=\{1,2,\cdots,M\}\) be the random binning.
There exists a constant \(C_\alpha<\infty\), depending only on
\(\alpha\), such that
\begin{align}
\mathbb E_FQ_\alpha\left(
\mathcal R_F(\rho_{XE})
\middle\|
\frac{\mathbbm{1}_\mathcal{Z}}{|\mathcal{Z}|}\otimes\rho_E
\right)-1
\le
C_\alpha\Big[
2^{-\eta'(\alpha)}
+2^{-\frac{1}{2}\eta'(\alpha)
-\frac{1}{2}\eta(\alpha-1)}
+2^{1-\eta(\alpha-1)}
\Big]
\label{eq:PA-upper}
\end{align}
and
\begin{align}
\mathbb E_FQ_\alpha\left(
\mathcal R_F(\rho_{XE})
\middle\|
\frac{\mathbbm{1}_\mathcal{Z}}{|\mathcal{Z}|}\otimes\rho_E
\right)-1
\ge&\
\max\Big\{
2^{-1}
2^{-\eta'(\alpha)},
2^{1-2\alpha}
2^{-\eta(\alpha-1)}
\Big\},
\label{eq:PA-lower}
\end{align}
where $\eta(t):=t(H_{1+t}\left(X|E\right)_{\rho_{XE}}-R)$ and $\eta'(t):=H_{2}^{(t)}(X|E)_{\rho_{XE}}-R$.
\end{theorem}
\begin{proof}
For each \(z\in\{1,2,\cdots,M\}\), set
\begin{equation*}
S_z(F):=\sum_{x:F(x)=z}A_x.
\end{equation*}
Blockwise functional calculus gives the exact identity
\begin{align*}
Q_\alpha\left(
\mathcal R_F(\rho_{XE})
\middle\|
\frac{\mathbbm{1}_\mathcal{Z}}{|\mathcal{Z}|}\otimes\rho_E
\right)
=
M^{\alpha-1}
\sum_{z=1}^M\operatorname{Tr} [S_z(F)]^\alpha.
\end{align*}
Fix \(z=1\), put
\begin{equation*}
\xi_x:=1_{\{F(x)=1\}},
\qquad
S:=\sum_x\xi_xA_x.
\end{equation*}
The random binning implies that the \(\xi_x\)'s are independent
Bernoulli-\(q\) variables.  By symmetry of the \(M\) bins, we have
\begin{align*}
\mathbb E_FQ_\alpha\left(
\mathcal R_F(\rho_{XE})
\middle\|
\frac{\mathbbm{1}_\mathcal{Z}}{|\mathcal{Z}|}\otimes\rho_E
\right)
&=
M^\alpha\mathbb E\operatorname{Tr} S^\alpha.
\end{align*}
Since \(\mathbb E S=\frac{1}{M}\rho_E^{1/\alpha}\) and
\(\operatorname{Tr}\rho_E=1\),
\begin{equation*}
\Tr (\mathbb E  S)^\alpha=M^{-\alpha}.
\end{equation*}
It follows that
\begin{align*}
\mathbb E_FQ_\alpha\left(
\mathcal R_F(\rho_{XE})
\middle\|
\frac{\mathbbm{1}_\mathcal{Z}}{|\mathcal{Z}|}\otimes\rho_E
\right)-1
=
M^{\alpha}
\left[
\mathbb E\operatorname{Tr} S^\alpha-\operatorname{Tr}(\mathbb{E}S)^\alpha
\right].
\end{align*}
Next we will apply Proposition~\ref{prop:Bernoulli}.  First,
\begin{equation*}
M^{\alpha}M^{1-\alpha}(1-M^{-1})
=
M(1-M^{-1})
=
M-1\leq M.
\end{equation*}
Second,
\begin{equation*}
M^{\alpha}[M^{-1}(1-M^{-1})]^{\alpha/2}
=
(M-1)^{\alpha/2}\leq M^{\alpha/2}.
\end{equation*}
Third,
\begin{align*}
M^{\alpha}(M^{-1}(1-M^{-1})^\alpha+(1-M^{-1})M^{-\alpha})
=
\frac{(M-1)^\alpha+M-1}{M},
\end{align*}
and
\begin{equation*}
2^{1-\alpha}M^{\alpha-1}\leq\frac{(M-1)^\alpha+M-1}{M}\leq 2M^{\alpha-1}
\end{equation*}
Fourth,
\begin{equation*}
M^{\alpha}(\alpha-1)M^{1-\alpha}(1-M^{-1})=(\alpha-1)(M-1)\geq \frac{M}{2}.
\end{equation*}
From Proposition~\ref{prop:Bernoulli}, substituting \eqref{eq:Theta-entropy} and
\eqref{eq:Xi-entropy} proves \eqref{eq:PA-upper} and
\eqref{eq:PA-lower}.
\end{proof}

\begin{proof}[Proof of Theorem~\ref{thm:10}:]
For any $n\in\mathbb{N}$,
let $F_n:\mathcal{X}^n \to \mathcal{Z}_n=\{1,2,\cdots,2^{nR}\}$ be
the random binning function. The operators in \eqref{eq:Ax-PA}
tensorize
\begin{equation*}
A_{x^n}
=
A_{x_1}\otimes\cdots\otimes A_{x_n},
\qquad
\overline A_n=\overline A^{\otimes n}.
\end{equation*}
Making use of the additivity of sandwiched R{\'e}nyi conditional entropy and the mixed-order order-two R{\'e}nyi conditional entropy, we get
\begin{align*}
\sum_{x^n}\operatorname{Tr} A_{x^n}^\alpha
=&
\left(\sum_x\operatorname{Tr} A_x^\alpha\right)^n
=
2^{-n(\alpha-1)H_\alpha(X|E)_{\rho_{XE}}},
\\
\sum_{x^n}
\operatorname{Tr} A_{x^n}^2\overline A_n^{\alpha-2}
=&
\left(\sum_x\operatorname{Tr} A_x^2\overline A^{\alpha-2}\right)^n
=
2^{-nH_2^{(\alpha)}(X|E)_{\rho_{XE}}}.
\end{align*}
Applying Theorem~\ref{thm:one-shot} with
\(F=F_n\) and \(\rho_{XE}=\rho_{XE}^{\otimes n}\), we find that there
exists a finite constant \(C_\alpha\), depending only on \(\alpha\), such that
\begin{align}
&\mathbb E_{F_n}
Q_\alpha\left(
\mathcal R_{F_n}(\rho_{XE}^{\otimes n})
\middle\|
\frac{\mathbbm{1}_{\mathcal{Z}_n}}{|\mathcal{Z}_n|}\otimes\rho_E^{\otimes n}
\right)-1 \nonumber\\
\le&
C_\alpha\Big[
2^{-n\eta'(\alpha)}
+
2^{-\frac{1}{2}n\eta'(\alpha)
-\frac{1}{2}n\eta(\alpha-1)}
+2^{1-n\eta(\alpha-1)}
\Big] \nonumber\\
\leq&3C_\alpha\max\left\{2^{-n\eta'(\alpha)},2^{-n\eta(\alpha-1)}\right\},
\label{eq:Delta-upper-exp}
\end{align}
where the last inequality follows from 
\begin{equation*}
2^{-\frac{1}{2}n\eta'(\alpha)
-\frac{1}{2}n\eta(\alpha-1)}\leq \max\left\{2^{-n\eta'(\alpha)},2^{-n\eta(\alpha-1)}\right\}.
\end{equation*}
The one-shot lower bound \eqref{eq:PA-lower} gives
\begin{align}
&\mathbb E_{F_n}
Q_\alpha\left(
\mathcal R_{F_n}(\rho_{XE}^{\otimes n})
\middle\|
\frac{\mathbbm{1}_{\mathcal{Z}_n}}{|\mathcal{Z}_n|}\otimes\rho_E^{\otimes n}
\right)-1 \nonumber\\
\ge&
\max\Big\{
2^{-1}\cdot
2^{-n\eta'(\alpha)},
2^{1-2\alpha}\cdot
2^{-n\eta(\alpha-1)}
\Big\}.\label{eq:Delta-lower-exp}
\end{align}
Since $H_{2}^{(\alpha)}(X|E)_{\rho_{XE}}\geq H_\alpha(X|E)_{\rho_{XE}}$ (cf. Corollary~\ref{cor:additivity-ordering-information}), we have
\begin{equation}\label{key}
\eta'(\alpha)>0, \quad \eta(\alpha-1)>0.
\end{equation}
The combination of \eqref{eq:Delta-upper-exp} and \eqref{eq:Delta-lower-exp} yields
\begin{align*}
&\lim_{n\to\infty}
-\frac{1}{n}\log \frac{1}{\alpha-1}\log\mathbb E_{F_n}
Q_\alpha\left(
\mathcal R_{F_n}(\rho_{XE}^{\otimes n})
\middle\|
\frac{\mathbbm{1}_{\mathcal{Z}_n}}{|\mathcal{Z}_n|}\otimes\rho_E^{\otimes n}
\right)\nonumber\\
=&\lim_{n\to\infty}
-\frac{1}{n}\log \max\Big\{
2^{-n\eta'(\alpha)},
2^{-n\eta(\alpha-1)}
\Big\} \nonumber\\
=&\min\{\eta'(\alpha),\eta(\alpha-1)\},
\end{align*}
where the first equation follows from \eqref{key} and the fact that
\(\log(1+f(n))\sim f(n)/\ln 2\) as \(f(n)\searrow0\). This completes the proof.
\end{proof}

\section{Conclusion and Discussion}\label{sec:dis}
In this work, using the associated mixed-order order-two mutual information and conditional entropy, we have characterized the reliability functions of quantum soft covering and quantum privacy amplification under the sandwiched R{\'e}nyi divergence of order $\alpha\in[2,\infty)$. Our results provide new operational meanings for both the sandwiched and mixed-order R{\'e}nyi information quantities. 

The sandwiched R{\'e}nyi quantities govern the higher-order contribution, whereas the mixed-order order-two quantities capture the quadratic fluctuation arising from random codebook sampling and random binning. Accounting for the interplay between these two contributions enables us to derive matching achievability and converse bounds and hence obtain the exact reliability functions. In particular, this yields, to the best of our knowledge, the first exact characterization of the reliability function for quantum soft covering.

We conclude by outlining several open problems.
\begin{enumerate}
\item  First, our analysis is restricted to the sandwiched Rényi divergence with orders $\alpha\in[2,\infty)$, while the reliability functions for $\alpha\in(0,2)$ remain unknown. For quantum soft covering, the method developed in \cite{LYH2023tight} yields a candidate reliability-function expression that exhibits a critical rate. This is in sharp contrast to the reliability function of classical soft covering, which has no critical rate~\cite{YuTan2018renyi}. An exact characterization over $\alpha\in(0,2)$ is therefore needed to determine whether this critical-rate phenomenon is intrinsic to noncommutativity or merely an artifact of the existing proof technique.

\item Second, classical-quantum channel resolvability can be viewed as a generalization of quantum soft covering. Reference \cite{hayashi2025resolvability} established the resolvability of classical-quantum channels under the trace distance. It is also important to determine the corresponding resolvability of classical-quantum channels under the sandwiched Rényi divergence of order $\alpha\in[2,\infty)$. Beyond the corresponding resolvability, it is also natural to investigate and characterize the exact reliability function in this regime.

\item Moreover, another important direction is to determine the exact reliability functions for quantum soft covering and quantum privacy amplification under the trace distance. Such characterizations would complement our results for the sandwiched Rényi divergence and further clarify the relationship between these two distinguishability measures at the level of error exponents. However, quite different techniques may be needed in the trace distance setting.

\item Finally, it would be interesting to explore applications of the mixed-order order-two Rényi divergence beyond quantum soft covering and privacy amplification. A closely related direction is to seek further generalizations of Lieb's trace functional and investigate their fundamental analytic properties. Such extensions may lead to new Rényi information quantities and reveal additional operational interpretations in quantum information theory.

\end{enumerate}

\appendix
\section{Auxiliary Lemmas}
This appendix contains several technical lemmas that are used in the proofs. 
\begin{lemma}[Complex interpolation theorem]
\label{lem:complex-interpolation-schatten}
Let
\[
\mathbb S:=\{z\in\mathbb C:0\leq \operatorname{Re}z\leq1\},
\]
and let \(F:\mathbb S\to\mathcal L(\mathcal H)\) be bounded and
continuous on \(\mathbb S\), and analytic in its interior. Let
\(1\leq p_0,p_1\leq\infty\), and define \(p_\theta\) by
\[
\frac{1}{p_\theta}
=
\frac{1-\theta}{p_0}
+
\frac{\theta}{p_1},
\qquad 0\leq\theta\leq1.
\]
Then
\[
\left\|F(\theta)\right\|_{p_\theta}
\leq
\left(
\sup_{y\in\mathbb R}\left\|F(iy)\right\|_{p_0}
\right)^{1-\theta}
\left(
\sup_{y\in\mathbb R}\left\|F(1+iy)\right\|_{p_1}
\right)^\theta.
\]
\end{lemma}
\begin{proof}
By the standard interpolation theorem, e.g., see \cite{BerghLofstrom1976}, we get
\[
\left\|F(\theta)\right\|_{[S_{p_0},S_{p_1}]_{\theta}}
\leq
\left(
\sup_{y\in\mathbb R}\left\|F(iy)\right\|_{p_0}
\right)^{1-\theta}
\left(
\sup_{y\in\mathbb R}\left\|F(1+iy)\right\|_{p_1}
\right)^\theta.
\]
Here $[S_{p_0},S_{p_1}]_{\theta}$ denotes the interpolation space. Moreover, \cite{kosaki1984} shows that
$$ \left\|F(\theta)\right\|_{[S_{p_0},S_{p_1}]_{\theta}}=\left\|F(\theta)\right\|_{p_\theta}.$$
\end{proof}

\begin{lemma}\label{w-young}
Let $a,b\geq 0$ and $0< \theta \leq 1$. For $\varepsilon>0$, we have
$$ a^{1-\theta}b^{\theta}\leq  \varepsilon a+C_{\theta,\varepsilon}b.$$ 
Here $C_{\theta,\varepsilon}<\infty$ depends on $\theta$ and $\varepsilon$.
\end{lemma}
\begin{proof}
When $\theta=1$ it is trivial. We only consider $0<\theta<1$. The celebrated Young's inequality gives
$$ x^{1-\theta}y^{\theta}\leq  (1-\theta)x+\theta y,\quad x,y\geq0.$$
Then applying it with
\[
x=\frac{\varepsilon}{1-\theta}a,
\qquad
y=
\left(\frac{1-\theta}{\varepsilon}\right)^{
\frac{1-\theta}{\theta}}b
\]
yields the assertion.
\end{proof}

\begin{lemma}[Clarkson--McCarthy inequality~\cite{hirzallah2002non}]\label{lem:CM}
Let \(2\le p<\infty\), \(A,B\in\mathcal{L}(\mathcal{H})\) and $\|A\|_p,\|B\|_p< +\infty$. Then
\begin{equation*}
2\left(\|A\|_p^p+\|B\|_p^p\right)
\le
\|A+B\|_p^p+\|A-B\|_p^p
\le
2^{p-1}\left(\|A\|_p^p+\|B\|_p^p\right).
\end{equation*}
\end{lemma}

\begin{lemma}[Birman--Koplienko--Solomyak inequality~\cite{BKS75,BS89}]\label{BKS}
Let \(0<\theta<1\) and \(1\le\alpha<\infty\). Let
\(A,B\in\mathcal{P}(\mathcal{H})\) be positive semidefinite operators
such that $
\norm{A-B}_{\alpha}<\infty.$
Then
\begin{equation*}
\norm{A^\theta-B^\theta}_{\alpha/\theta}
\le C_\theta
\norm{A-B}_{\alpha}^{\theta}.
\end{equation*}
\end{lemma}
\bmhead{Acknowledgements}
The research of Shi-Bing Li was supported by NSFC under Grant 62571166. Hongsen Qiu and Xinyu Zhang were supported
by the National Natural Science Foundation of China (Grant No. 12371138 and No. W2441002).
\bmhead{Data Availability} Data sharing is not applicable to this article, since no data sets were used.
\bibliography{references}


\begin{thebibliography}{69}
\ifx \bisbn   \undefined \def \bisbn  #1{ISBN #1}\fi
\ifx \binits  \undefined \def \binits#1{#1}\fi
\ifx \bauthor  \undefined \def \bauthor#1{#1}\fi
\ifx \batitle  \undefined \def \batitle#1{#1}\fi
\ifx \bjtitle  \undefined \def \bjtitle#1{#1}\fi
\ifx \bvolume  \undefined \def \bvolume#1{\textbf{#1}}\fi
\ifx \byear  \undefined \def \byear#1{#1}\fi
\ifx \bissue  \undefined \def \bissue#1{#1}\fi
\ifx \bfpage  \undefined \def \bfpage#1{#1}\fi
\ifx \blpage  \undefined \def \blpage #1{#1}\fi
\ifx \burl  \undefined \def \burl#1{\textsf{#1}}\fi
\ifx \doiurl  \undefined \def \doiurl#1{\url{https://doi.org/#1}}\fi
\ifx \betal  \undefined \def \betal{\textit{et al.}}\fi
\ifx \binstitute  \undefined \def \binstitute#1{#1}\fi
\ifx \binstitutionaled  \undefined \def \binstitutionaled#1{#1}\fi
\ifx \bctitle  \undefined \def \bctitle#1{#1}\fi
\ifx \beditor  \undefined \def \beditor#1{#1}\fi
\ifx \bpublisher  \undefined \def \bpublisher#1{#1}\fi
\ifx \bbtitle  \undefined \def \bbtitle#1{#1}\fi
\ifx \bedition  \undefined \def \bedition#1{#1}\fi
\ifx \bseriesno  \undefined \def \bseriesno#1{#1}\fi
\ifx \blocation  \undefined \def \blocation#1{#1}\fi
\ifx \bsertitle  \undefined \def \bsertitle#1{#1}\fi
\ifx \bsnm \undefined \def \bsnm#1{#1}\fi
\ifx \bsuffix \undefined \def \bsuffix#1{#1}\fi
\ifx \bparticle \undefined \def \bparticle#1{#1}\fi
\ifx \barticle \undefined \def \barticle#1{#1}\fi
\bibcommenthead
\ifx \bconfdate \undefined \def \bconfdate #1{#1}\fi
\ifx \botherref \undefined \def \botherref #1{#1}\fi
\ifx \url \undefined \def \url#1{\textsf{#1}}\fi
\ifx \bchapter \undefined \def \bchapter#1{#1}\fi
\ifx \bbook \undefined \def \bbook#1{#1}\fi
\ifx \bcomment \undefined \def \bcomment#1{#1}\fi
\ifx \oauthor \undefined \def \oauthor#1{#1}\fi
\ifx \citeauthoryear \undefined \def \citeauthoryear#1{#1}\fi
\ifx \endbibitem  \undefined \def \endbibitem {}\fi
\ifx \bconflocation  \undefined \def \bconflocation#1{#1}\fi
\ifx \arxivurl  \undefined \def \arxivurl#1{\textsf{#1}}\fi
\csname PreBibitemsHook\endcsname

\bibitem[\protect\citeauthoryear{Lieb}{1973}]{lieb1973convex}
\begin{barticle}
\bauthor{\bsnm{Lieb}, \binits{E.H.}}:
\batitle{Convex trace functions and the {Wigner-Yanase-Dyson} conjecture}.
\bjtitle{Advances in Mathematics}
\bvolume{11}(\bissue{3}),
\bfpage{267}--\blpage{288}
(\byear{1973})
\end{barticle}
\endbibitem

\bibitem[\protect\citeauthoryear{Ando}{1979}]{ando1979concavity}
\begin{barticle}
\bauthor{\bsnm{Ando}, \binits{T.}}:
\batitle{Concavity of certain maps on positive definite matrices and applications to hadamard products}.
\bjtitle{Linear algebra and its applications}
\bvolume{26},
\bfpage{203}--\blpage{241}
(\byear{1979})
\end{barticle}
\endbibitem

\bibitem[\protect\citeauthoryear{Petz}{1986}]{petz1986quasi}
\begin{barticle}
\bauthor{\bsnm{Petz}, \binits{D.}}:
\batitle{Quasi-entropies for finite quantum systems}.
\bjtitle{Reports on mathematical physics}
\bvolume{23}(\bissue{1}),
\bfpage{57}--\blpage{65}
(\byear{1986})
\end{barticle}
\endbibitem

\bibitem[\protect\citeauthoryear{Jen{\v c}ov{\'a} and Ruskai}{2010}]{jencova2010unified}
\begin{barticle}
\bauthor{\bsnm{Jen{\v c}ov{\'a}}, \binits{A.}},
\bauthor{\bsnm{Ruskai}, \binits{M.B.}}:
\batitle{A unified treatment of convexity of relative entropy and related trace functions, with conditions for equality}.
\bjtitle{Reviews in Mathematical Physics}
\bvolume{22}(\bissue{9}),
\bfpage{1099}--\blpage{1121}
(\byear{2010})
\end{barticle}
\endbibitem

\bibitem[\protect\citeauthoryear{Hiai}{2013}]{hiai2013concavity}
\begin{barticle}
\bauthor{\bsnm{Hiai}, \binits{F.}}:
\batitle{Concavity of certain matrix trace and norm functions}.
\bjtitle{Linear algebra and its applications}
\bvolume{439}(\bissue{5}),
\bfpage{1568}--\blpage{1589}
(\byear{2013})
\end{barticle}
\endbibitem

\bibitem[\protect\citeauthoryear{Carlen et~al.}{2016}]{carlen2016some}
\begin{barticle}
\bauthor{\bsnm{Carlen}, \binits{E.A.}},
\bauthor{\bsnm{Frank}, \binits{R.L.}},
\bauthor{\bsnm{Lieb}, \binits{E.H.}}:
\batitle{Some operator and trace function convexity theorems}.
\bjtitle{Linear algebra and its applications}
\bvolume{490},
\bfpage{174}--\blpage{185}
(\byear{2016})
\end{barticle}
\endbibitem

\bibitem[\protect\citeauthoryear{Zhang}{2020}]{zhang2020wigner}
\begin{barticle}
\bauthor{\bsnm{Zhang}, \binits{H.}}:
\batitle{From wigner-yanase-dyson conjecture to carlen-frank-lieb conjecture}.
\bjtitle{Advances in Mathematics}
\bvolume{365},
\bfpage{107053}
(\byear{2020})
\end{barticle}
\endbibitem

\bibitem[\protect\citeauthoryear{Epstein}{1973}]{epstein1973remarks}
\begin{barticle}
\bauthor{\bsnm{Epstein}, \binits{H.}}:
\batitle{Remarks on two theorems of e. lieb}.
\bjtitle{Communications in Mathematical Physics}
\bvolume{31}(\bissue{4}),
\bfpage{317}--\blpage{325}
(\byear{1973})
\end{barticle}
\endbibitem

\bibitem[\protect\citeauthoryear{Frank and Lieb}{2013}]{frank2013monotonicity}
\begin{botherref}
\oauthor{\bsnm{Frank}, \binits{R.L.}},
\oauthor{\bsnm{Lieb}, \binits{E.H.}}:
Monotonicity of a relative {R{\'e}nyi} entropy.
Journal of Mathematical Physics
\textbf{54}(12)
(2013)
\end{botherref}
\endbibitem

\bibitem[\protect\citeauthoryear{Beigi}{2013}]{Beigi2013sandwiched}
\begin{barticle}
\bauthor{\bsnm{Beigi}, \binits{S.}}:
\batitle{Sandwiched {{R{\'e}nyi}} divergence satisfies data processing inequality}.
\bjtitle{J. Math. Phys.}
\bvolume{54},
\bfpage{122202}
(\byear{2013})
\end{barticle}
\endbibitem

\bibitem[\protect\citeauthoryear{Audenaert and Datta}{2015}]{audenaert2015alpha}
\begin{botherref}
\oauthor{\bsnm{Audenaert}, \binits{K.M.}},
\oauthor{\bsnm{Datta}, \binits{N.}}:
$\alpha$-z {R{\'e}nyi} relative entropies.
Journal of Mathematical Physics
\textbf{56}(2)
(2015)
\end{botherref}
\endbibitem

\bibitem[\protect\citeauthoryear{M{\"u}ller-Lennert et~al.}{2013}]{muller2013quantum}
\begin{botherref}
\oauthor{\bsnm{M{\"u}ller-Lennert}, \binits{M.}},
\oauthor{\bsnm{Dupuis}, \binits{F.}},
\oauthor{\bsnm{Szehr}, \binits{O.}},
\oauthor{\bsnm{Fehr}, \binits{S.}},
\oauthor{\bsnm{Tomamichel}, \binits{M.}}:
On quantum {R{\'e}nyi} entropies: A new generalization and some properties.
Journal of Mathematical Physics
\textbf{54}(12)
(2013)
\end{botherref}
\endbibitem

\bibitem[\protect\citeauthoryear{Wilde et~al.}{2014}]{wilde2014strong}
\begin{barticle}
\bauthor{\bsnm{Wilde}, \binits{M.M.}},
\bauthor{\bsnm{Winter}, \binits{A.}},
\bauthor{\bsnm{Yang}, \binits{D.}}:
\batitle{Strong converse for the classical capacity of entanglement-breaking and hadamard channels via a sandwiched {R{\'e}nyi} relative entropy}.
\bjtitle{Communications in Mathematical Physics}
\bvolume{331}(\bissue{2}),
\bfpage{593}--\blpage{622}
(\byear{2014})
\end{barticle}
\endbibitem

\bibitem[\protect\citeauthoryear{Audenaert et~al.}{2007}]{AKCMBMA2007discriminating}
\begin{barticle}
\bauthor{\bsnm{Audenaert}, \binits{K.M.}},
\bauthor{\bsnm{Calsamiglia}, \binits{J.}},
\bauthor{\bsnm{Munoz-Tapia}, \binits{R.}},
\bauthor{\bsnm{Bagan}, \binits{E.}},
\bauthor{\bsnm{Masanes}, \binits{L.}},
\bauthor{\bsnm{Acin}, \binits{A.}},
\bauthor{\bsnm{Verstraete}, \binits{F.}}:
\batitle{Discriminating states: the quantum {Chernoff} bound}.
\bjtitle{Phys. Rev. Lett.}
\bvolume{98},
\bfpage{160501}
(\byear{2007})
\end{barticle}
\endbibitem

\bibitem[\protect\citeauthoryear{Hayashi}{2007}]{Hayashi2007error}
\begin{barticle}
\bauthor{\bsnm{Hayashi}, \binits{M.}}:
\batitle{Error exponent in asymmetric quantum hypothesis testing and its application to classical-quantum channel coding}.
\bjtitle{Phys. Rev. A}
\bvolume{76},
\bfpage{062301}
(\byear{2007})
\end{barticle}
\endbibitem

\bibitem[\protect\citeauthoryear{Audenaert et~al.}{2008}]{ANSV2008asymptotic}
\begin{barticle}
\bauthor{\bsnm{Audenaert}, \binits{K.M.}},
\bauthor{\bsnm{Nussbaum}, \binits{M.}},
\bauthor{\bsnm{Szko{\l}a}, \binits{A.}},
\bauthor{\bsnm{Verstraete}, \binits{F.}}:
\batitle{Asymptotic error rates in quantum hypothesis testing}.
\bjtitle{Commun. Math. Phys.}
\bvolume{279}(\bissue{1}),
\bfpage{251}--\blpage{283}
(\byear{2008})
\end{barticle}
\endbibitem

\bibitem[\protect\citeauthoryear{Hayashi and Tomamichel}{2016}]{HayashiTomamichel2016correlation}
\begin{barticle}
\bauthor{\bsnm{Hayashi}, \binits{M.}},
\bauthor{\bsnm{Tomamichel}, \binits{M.}}:
\batitle{Correlation detection and an operational interpretation of the {R{\'e}nyi} mutual information}.
\bjtitle{J. Math. Phys.}
\bvolume{57},
\bfpage{102201}
(\byear{2016})
\end{barticle}
\endbibitem

\bibitem[\protect\citeauthoryear{Cheng et~al.}{2021}]{CHDH2021non}
\begin{barticle}
\bauthor{\bsnm{Cheng}, \binits{H.-C.}},
\bauthor{\bsnm{Hanson}, \binits{E.P.}},
\bauthor{\bsnm{Datta}, \binits{N.}},
\bauthor{\bsnm{Hsieh}, \binits{M.-H.}}:
\batitle{Non-asymptotic classical data compression with quantum side information}.
\bjtitle{IEEE Trans. Inf. Theory}
\bvolume{67}(\bissue{2}),
\bfpage{902}--\blpage{930}
(\byear{2021})
\end{barticle}
\endbibitem

\bibitem[\protect\citeauthoryear{Dalai}{2013}]{Dalai2013lower}
\begin{barticle}
\bauthor{\bsnm{Dalai}, \binits{M.}}:
\batitle{Lower bounds on the probability of error for classical and classical-quantum channels}.
\bjtitle{IEEE Trans. Inf. Theory}
\bvolume{59}(\bissue{12}),
\bfpage{8027}--\blpage{8056}
(\byear{2013})
\end{barticle}
\endbibitem

\bibitem[\protect\citeauthoryear{Cheng et~al.}{2019}]{CHT2019quantum}
\begin{barticle}
\bauthor{\bsnm{Cheng}, \binits{H.-C.}},
\bauthor{\bsnm{Hsieh}, \binits{M.-H.}},
\bauthor{\bsnm{Tomamichel}, \binits{M.}}:
\batitle{Quantum sphere-packing bounds with polynomial prefactors}.
\bjtitle{IEEE Trans. Inf. Theory}
\bvolume{65}(\bissue{5}),
\bfpage{2872}--\blpage{2898}
(\byear{2019})
\end{barticle}
\endbibitem

\bibitem[\protect\citeauthoryear{Hahn et~al.}{2024}]{HahnTanBrown2024}
\begin{botherref}
\oauthor{\bsnm{Hahn}, \binits{T.A.}},
\oauthor{\bsnm{Tan}, \binits{E.Y.-Z.}},
\oauthor{\bsnm{Brown}, \binits{P.}}:
Bounds on {Petz}--{R}{\'e}nyi divergences and their applications for device-independent cryptography.
arXiv preprint arXiv:2408.12313
(2024)
{\href{https://arxiv.org/abs/2408.12313}{{arXiv:2408.12313}}}
{[quant-ph]}
\end{botherref}
\endbibitem

\bibitem[\protect\citeauthoryear{Renes}{2024}]{renes2024tight}
\begin{barticle}
\bauthor{\bsnm{Renes}, \binits{J.M.}}:
\batitle{Tight lower bound on the error exponent of classical-quantum channels}.
\bjtitle{IEEE Transactions on Information Theory}
\bvolume{71}(\bissue{1}),
\bfpage{530}--\blpage{538}
(\byear{2024})
\end{barticle}
\endbibitem

\bibitem[\protect\citeauthoryear{Li and Yang}{2025}]{LiYang2025}
\begin{barticle}
\bauthor{\bsnm{Li}, \binits{K.}},
\bauthor{\bsnm{Yang}, \binits{D.}}:
\batitle{Reliability function of classical--quantum channels}.
\bjtitle{Physical Review Letters}
\bvolume{134}(\bissue{1}),
\bfpage{010802}
(\byear{2025})
\end{barticle}
\endbibitem

\bibitem[\protect\citeauthoryear{Mosonyi and Ogawa}{2015}]{MosonyiOgawa2015quantum}
\begin{barticle}
\bauthor{\bsnm{Mosonyi}, \binits{M.}},
\bauthor{\bsnm{Ogawa}, \binits{T.}}:
\batitle{Quantum hypothesis testing and the operational interpretation of the quantum {R{\'e}nyi} relative entropies}.
\bjtitle{Commun. Math. Phys.}
\bvolume{334}(\bissue{3}),
\bfpage{1617}--\blpage{1648}
(\byear{2015})
\end{barticle}
\endbibitem

\bibitem[\protect\citeauthoryear{Mosonyi and Ogawa}{2017}]{MosonyiOgawa2017strong}
\begin{barticle}
\bauthor{\bsnm{Mosonyi}, \binits{M.}},
\bauthor{\bsnm{Ogawa}, \binits{T.}}:
\batitle{Strong converse exponent for classical--quantum channel coding}.
\bjtitle{Commun. Math. Phys.}
\bvolume{355}(\bissue{1}),
\bfpage{373}--\blpage{426}
(\byear{2017})
\end{barticle}
\endbibitem

\bibitem[\protect\citeauthoryear{Gupta and Wilde}{2015}]{GuptaWilde2015multiplicativity}
\begin{barticle}
\bauthor{\bsnm{Gupta}, \binits{M.K.}},
\bauthor{\bsnm{Wilde}, \binits{M.M.}}:
\batitle{Multiplicativity of completely bounded $p$-norms implies a strong converse for entanglement-assisted capacity}.
\bjtitle{Commun. Math. Phys.}
\bvolume{334}(\bissue{2}),
\bfpage{867}--\blpage{887}
(\byear{2015})
\end{barticle}
\endbibitem

\bibitem[\protect\citeauthoryear{Cooney et~al.}{2016}]{CMW2016strong}
\begin{barticle}
\bauthor{\bsnm{Cooney}, \binits{T.}},
\bauthor{\bsnm{Mosonyi}, \binits{M.}},
\bauthor{\bsnm{Wilde}, \binits{M.M.}}:
\batitle{Strong converse exponents for a quantum channel discrimination problem and quantum-feedback-assisted communication}.
\bjtitle{Commun. Math. Phys.}
\bvolume{344}(\bissue{3}),
\bfpage{797}--\blpage{829}
(\byear{2016})
\end{barticle}
\endbibitem

\bibitem[\protect\citeauthoryear{Li and Yao}{2024}]{LiYao2024}
\begin{barticle}
\bauthor{\bsnm{Li}, \binits{K.}},
\bauthor{\bsnm{Yao}, \binits{Y.}}:
\batitle{Operational interpretation of the sandwiched r{\'e}nyi divergence of order 1/2 to 1 as strong converse exponents}.
\bjtitle{Communications in Mathematical Physics}
\bvolume{405}(\bissue{2}),
\bfpage{22}
(\byear{2024})
\end{barticle}
\endbibitem

\bibitem[\protect\citeauthoryear{Li et~al.}{2023}]{LYH2023tight}
\begin{barticle}
\bauthor{\bsnm{Li}, \binits{K.}},
\bauthor{\bsnm{Yao}, \binits{Y.}},
\bauthor{\bsnm{Hayashi}, \binits{M.}}:
\batitle{Tight exponential analysis for smoothing the max-relative entropy and for quantum privacy amplification}.
\bjtitle{IEEE Trans. Inf. Theory}
\bvolume{69}(\bissue{3}),
\bfpage{1680}--\blpage{1694}
(\byear{2023})
\end{barticle}
\endbibitem

\bibitem[\protect\citeauthoryear{Li and Yao}{2024}]{LiYao2024reliability}
\begin{barticle}
\bauthor{\bsnm{Li}, \binits{K.}},
\bauthor{\bsnm{Yao}, \binits{Y.}}:
\batitle{Reliability function of quantum information decoupling via the sandwiched {R\'enyi} divergence}.
\bjtitle{Commun. Math. Phys.}
\bvolume{405}(\bissue{7}),
\bfpage{160}
(\byear{2024})
\end{barticle}
\endbibitem

\bibitem[\protect\citeauthoryear{Li and Yao}{2025}]{LiYaoChannelSimulation2025}
\begin{barticle}
\bauthor{\bsnm{Li}, \binits{K.}},
\bauthor{\bsnm{Yao}, \binits{Y.}}:
\batitle{Reliable simulation of quantum channels: The error exponent}.
\bjtitle{IEEE Transactions on Information Theory}
\bvolume{71}(\bissue{1}),
\bfpage{518}--\blpage{529}
(\byear{2025})
\end{barticle}
\endbibitem

\bibitem[\protect\citeauthoryear{Berta et~al.}{2026}]{TightAnyShotDecoupling2026}
\begin{botherref}
\oauthor{\bsnm{Berta}, \binits{M.}},
\oauthor{\bsnm{Cheng}, \binits{H.-C.}},
\oauthor{\bsnm{Yao}, \binits{Y.}}:
Tight any-shot quantum decoupling.
arXiv preprint arXiv:2602.17430
(2026)
{\href{https://arxiv.org/abs/2602.17430}{{arXiv:2602.17430}}}
{[quant-ph]}
\end{botherref}
\endbibitem

\bibitem[\protect\citeauthoryear{Tomamichel}{2016}]{tomamichel2016quantum}
\begin{bbook}
\bauthor{\bsnm{Tomamichel}, \binits{M.}}:
\bbtitle{Quantum Information Processing with Finite Resources: Mathematical Foundations}.
\bsertitle{SpringerBriefs in Mathematical Physics},
vol. \bseriesno{5}.
\bpublisher{Springer},
\blocation{Cham}
(\byear{2016})
\end{bbook}
\endbibitem

\bibitem[\protect\citeauthoryear{Verhagen et~al.}{2026}]{VerhagenTomamichelHaapasalo2026}
\begin{barticle}
\bauthor{\bsnm{Verhagen}, \binits{F.}},
\bauthor{\bsnm{Tomamichel}, \binits{M.}},
\bauthor{\bsnm{Haapasalo}, \binits{E.}}:
\batitle{Conditions for large-sample majorization of pairs of flat states in terms of {$\alpha$-$z$} relative entropies}.
\bjtitle{Communications in Mathematical Physics}
\bvolume{407}(\bissue{7}),
\bfpage{157}
(\byear{2026})
\end{barticle}
\endbibitem

\bibitem[\protect\citeauthoryear{Rubboli et~al.}{2026}]{rubboli2024quantum}
\begin{barticle}
\bauthor{\bsnm{Rubboli}, \binits{R.}},
\bauthor{\bsnm{Goodarzi}, \binits{M.M.}},
\bauthor{\bsnm{Tomamichel}, \binits{M.}}:
\batitle{Quantum conditional entropies from convex trace functionals}.
\bjtitle{Communications in Mathematical Physics}
\bvolume{407}(\bissue{9}),
\bfpage{180}
(\byear{2026})
\end{barticle}
\endbibitem

\bibitem[\protect\citeauthoryear{Rubboli and Tomamichel}{2026}]{RubboliTomamichel2026}
\begin{botherref}
\oauthor{\bsnm{Rubboli}, \binits{R.}},
\oauthor{\bsnm{Tomamichel}, \binits{M.}}:
The Strong Converse Exponent of Composable Randomness Extraction Against Quantum Side Information
(2026)
\end{botherref}
\endbibitem

\bibitem[\protect\citeauthoryear{Lindblad}{1974}]{lindblad1974}
\begin{barticle}
\bauthor{\bsnm{Lindblad}, \binits{G.}}:
\batitle{Expectations and entropy inequalities for finite quantum systems}.
\bjtitle{Communications in Mathematical Physics}
\bvolume{39}(\bissue{2}),
\bfpage{111}--\blpage{119}
(\byear{1974})
\end{barticle}
\endbibitem

\bibitem[\protect\citeauthoryear{Lindblad}{1975}]{lindblad1975}
\begin{barticle}
\bauthor{\bsnm{Lindblad}, \binits{G.}}:
\batitle{Completely positive maps and entropy inequalities}.
\bjtitle{Communications in Mathematical Physics}
\bvolume{40}(\bissue{2}),
\bfpage{147}--\blpage{151}
(\byear{1975})
\end{barticle}
\endbibitem

\bibitem[\protect\citeauthoryear{Ahlswede and Winter}{2002}]{AhlswedeWinter2002}
\begin{barticle}
\bauthor{\bsnm{Ahlswede}, \binits{R.}},
\bauthor{\bsnm{Winter}, \binits{A.}}:
\batitle{Strong converse for identification via quantum channels}.
\bjtitle{IEEE Transactions on Information Theory}
\bvolume{48}(\bissue{3}),
\bfpage{569}--\blpage{579}
(\byear{2002})
\end{barticle}
\endbibitem

\bibitem[\protect\citeauthoryear{Bennett et~al.}{2014}]{bennettEtAl2014}
\begin{barticle}
\bauthor{\bsnm{Bennett}, \binits{C.H.}},
\bauthor{\bsnm{Devetak}, \binits{I.}},
\bauthor{\bsnm{Harrow}, \binits{A.W.}},
\bauthor{\bsnm{Shor}, \binits{P.W.}},
\bauthor{\bsnm{Winter}, \binits{A.}}:
\batitle{The quantum reverse shannon theorem and resource tradeoffs for simulating quantum channels}.
\bjtitle{IEEE Transactions on Information Theory}
\bvolume{60}(\bissue{5}),
\bfpage{2926}--\blpage{2959}
(\byear{2014})
\end{barticle}
\endbibitem

\bibitem[\protect\citeauthoryear{Devetak}{2005}]{devetak2005private}
\begin{barticle}
\bauthor{\bsnm{Devetak}, \binits{I.}}:
\batitle{The private classical capacity and quantum capacity of a quantum channel}.
\bjtitle{IEEE Transactions on Information Theory}
\bvolume{51}(\bissue{1}),
\bfpage{44}--\blpage{55}
(\byear{2005})
\end{barticle}
\endbibitem

\bibitem[\protect\citeauthoryear{Devetak and Winter}{2003}]{DevetakWinter2003}
\begin{barticle}
\bauthor{\bsnm{Devetak}, \binits{I.}},
\bauthor{\bsnm{Winter}, \binits{A.}}:
\batitle{Classical data compression with quantum side information}.
\bjtitle{Physical Review A}
\bvolume{68}(\bissue{4}),
\bfpage{042301}
(\byear{2003})
\end{barticle}
\endbibitem

\bibitem[\protect\citeauthoryear{Devetak and Winter}{2005}]{DevetakWinter2005}
\begin{barticle}
\bauthor{\bsnm{Devetak}, \binits{I.}},
\bauthor{\bsnm{Winter}, \binits{A.}}:
\batitle{Distillation of secret key and entanglement from quantum states}.
\bjtitle{Proceedings of the Royal Society A: Mathematical, Physical and Engineering Sciences}
\bvolume{461}(\bissue{2053}),
\bfpage{207}--\blpage{235}
(\byear{2005})
\end{barticle}
\endbibitem

\bibitem[\protect\citeauthoryear{Winter}{2005}]{Winter2005}
\begin{bchapter}
\bauthor{\bsnm{Winter}, \binits{A.}}:
\bctitle{Secret, public and quantum correlation cost of triples of random variables}.
In: \bbtitle{Proceedings of the 2005 IEEE International Symposium on Information Theory},
\bconflocation{Adelaide, Australia},
pp. \bfpage{2270}--\blpage{2274}
(\byear{2005})
\end{bchapter}
\endbibitem

\bibitem[\protect\citeauthoryear{Cai et~al.}{2004}]{CaiWinterYeung2004}
\begin{barticle}
\bauthor{\bsnm{Cai}, \binits{N.}},
\bauthor{\bsnm{Winter}, \binits{A.}},
\bauthor{\bsnm{Yeung}, \binits{R.W.}}:
\batitle{Quantum privacy and quantum wiretap channels}.
\bjtitle{Problems of Information Transmission}
\bvolume{40}(\bissue{4}),
\bfpage{318}--\blpage{336}
(\byear{2004})
\end{barticle}
\endbibitem

\bibitem[\protect\citeauthoryear{Tomamichel et~al.}{2011}]{TSSR2011leftover}
\begin{barticle}
\bauthor{\bsnm{Tomamichel}, \binits{M.}},
\bauthor{\bsnm{Schaffner}, \binits{C.}},
\bauthor{\bsnm{Smith}, \binits{A.}},
\bauthor{\bsnm{Renner}, \binits{R.}}:
\batitle{Leftover hashing against quantum side information}.
\bjtitle{IEEE Trans. Inf. Theory}
\bvolume{57}(\bissue{8}),
\bfpage{5524}--\blpage{5535}
(\byear{2011})
\end{barticle}
\endbibitem

\bibitem[\protect\citeauthoryear{Cheng and Gao}{2024}]{ChengGao2024}
\begin{barticle}
\bauthor{\bsnm{Cheng}, \binits{H.-C.}},
\bauthor{\bsnm{Gao}, \binits{L.}}:
\batitle{Error exponent and strong converse for quantum soft covering}.
\bjtitle{IEEE Transactions on Information Theory}
\bvolume{70}(\bissue{5}),
\bfpage{3499}--\blpage{3511}
(\byear{2024})
\end{barticle}
\endbibitem

\bibitem[\protect\citeauthoryear{Shen et~al.}{2024}]{ShenGaoCheng2024}
\begin{barticle}
\bauthor{\bsnm{Shen}, \binits{Y.-C.}},
\bauthor{\bsnm{Gao}, \binits{L.}},
\bauthor{\bsnm{Cheng}, \binits{H.-C.}}:
\batitle{Optimal second-order rates for quantum soft covering and privacy amplification}.
\bjtitle{IEEE Transactions on Information Theory}
\bvolume{70}(\bissue{7}),
\bfpage{5077}--\blpage{5091}
(\byear{2024})
\end{barticle}
\endbibitem

\bibitem[\protect\citeauthoryear{Renner and K{\"o}nig}{2005}]{RennerKonig2005}
\begin{bchapter}
\bauthor{\bsnm{Renner}, \binits{R.}},
\bauthor{\bsnm{K{\"o}nig}, \binits{R.}}:
\bctitle{Universally composable privacy amplification against quantum adversaries}.
In: \bbtitle{Theory of Cryptography}.
\bsertitle{Lecture Notes in Computer Science},
vol. \bseriesno{3378},
pp. \bfpage{407}--\blpage{425}.
\bpublisher{Springer},
\blocation{Berlin}
(\byear{2005})
\end{bchapter}
\endbibitem

\bibitem[\protect\citeauthoryear{Renner}{2005}]{Renner2005}
\begin{botherref}
\oauthor{\bsnm{Renner}, \binits{R.}}:
Security of quantum key distribution.
PhD thesis,
ETH Z{\"u}rich
(2005)
\end{botherref}
\endbibitem

\bibitem[\protect\citeauthoryear{Hayashi}{2015}]{Hayashi2015}
\begin{barticle}
\bauthor{\bsnm{Hayashi}, \binits{M.}}:
\batitle{Precise evaluation of leaked information with secure randomness extraction in the presence of quantum attacker}.
\bjtitle{Communications in Mathematical Physics}
\bvolume{333}(\bissue{1}),
\bfpage{335}--\blpage{350}
(\byear{2015})
\end{barticle}
\endbibitem

\bibitem[\protect\citeauthoryear{M{\"u}ller-Lennert et~al.}{2013}]{MDSFT2013on}
\begin{barticle}
\bauthor{\bsnm{M{\"u}ller-Lennert}, \binits{M.}},
\bauthor{\bsnm{Dupuis}, \binits{F.}},
\bauthor{\bsnm{Szehr}, \binits{O.}},
\bauthor{\bsnm{Fehr}, \binits{S.}},
\bauthor{\bsnm{Tomamichel}, \binits{M.}}:
\batitle{On quantum {R{\'e}nyi} entropies: a new generalization and some properties}.
\bjtitle{J. Math. Phys.}
\bvolume{54},
\bfpage{122203}
(\byear{2013})
\end{barticle}
\endbibitem

\bibitem[\protect\citeauthoryear{Wilde et~al.}{2014}]{WWY2014strong}
\begin{barticle}
\bauthor{\bsnm{Wilde}, \binits{M.M.}},
\bauthor{\bsnm{Winter}, \binits{A.}},
\bauthor{\bsnm{Yang}, \binits{D.}}:
\batitle{Strong converse for the classical capacity of entanglement-breaking and {Hadamard} channels via a sandwiched {R{\'e}nyi} relative entropy}.
\bjtitle{Commun. Math. Phys.}
\bvolume{331}(\bissue{2}),
\bfpage{593}--\blpage{622}
(\byear{2014})
\end{barticle}
\endbibitem

\bibitem[\protect\citeauthoryear{Umegaki}{1954}]{Umegaki1954conditional}
\begin{barticle}
\bauthor{\bsnm{Umegaki}, \binits{H.}}:
\batitle{Conditional expectation in an operator algebra}.
\bjtitle{Tohoku Math. J.}
\bvolume{6}(\bissue{2}),
\bfpage{177}--\blpage{181}
(\byear{1954})
\end{barticle}
\endbibitem

\bibitem[\protect\citeauthoryear{Tomamichel et~al.}{2014}]{TBH2014relating}
\begin{barticle}
\bauthor{\bsnm{Tomamichel}, \binits{M.}},
\bauthor{\bsnm{Berta}, \binits{M.}},
\bauthor{\bsnm{Hayashi}, \binits{M.}}:
\batitle{Relating different quantum generalizations of the conditional {R{\'e}nyi} entropy}.
\bjtitle{J. Math. Phys.}
\bvolume{55},
\bfpage{082206}
(\byear{2014})
\end{barticle}
\endbibitem

\bibitem[\protect\citeauthoryear{Temme et~al.}{2010}]{chi-diver}
\begin{barticle}
\bauthor{\bsnm{Temme}, \binits{K.}},
\bauthor{\bsnm{Kastoryano}, \binits{M.J.}},
\bauthor{\bsnm{Ruskai}, \binits{M.B.}},
\bauthor{\bsnm{Wolf}, \binits{M.M.}},
\bauthor{\bsnm{Verstraete}, \binits{F.}}:
\batitle{The {$\chi^2$}-divergence and mixing times of quantum markov processes}.
\bjtitle{Journal of Mathematical Physics}
\bvolume{51}(\bissue{12}),
\bfpage{122201}
(\byear{2010})
\end{barticle}
\endbibitem

\bibitem[\protect\citeauthoryear{Carlen and Maas}{2017}]{CM17}
\begin{barticle}
\bauthor{\bsnm{Carlen}, \binits{E.A.}},
\bauthor{\bsnm{Maas}, \binits{J.}}:
\batitle{Gradient flow and entropy inequalities for quantum markov semigroups with detailed balance}.
\bjtitle{Journal of Functional Analysis}
\bvolume{273}(\bissue{5}),
\bfpage{1810}--\blpage{1869}
(\byear{2017})
\end{barticle}
\endbibitem

\bibitem[\protect\citeauthoryear{Araki}{1990}]{Araki1990inequality}
\begin{barticle}
\bauthor{\bsnm{Araki}, \binits{H.}}:
\batitle{On an inequality of {Lieb} and {Thirring}}.
\bjtitle{Letters in Mathematical Physics}
\bvolume{19},
\bfpage{167}--\blpage{170}
(\byear{1990})
\end{barticle}
\endbibitem

\bibitem[\protect\citeauthoryear{Junge and Xu}{2008}]{JungeXu2008}
\begin{barticle}
\bauthor{\bsnm{Junge}, \binits{M.}},
\bauthor{\bsnm{Xu}, \binits{Q.}}:
\batitle{Noncommutative {B}urkholder/{R}osenthal inequalities {II}: Applications}.
\bjtitle{Israel Journal of Mathematics}
\bvolume{167},
\bfpage{227}--\blpage{282}
(\byear{2008})
\end{barticle}
\endbibitem

\bibitem[\protect\citeauthoryear{Pisier and Xu}{2003}]{PX03}
\begin{bchapter}
\bauthor{\bsnm{Pisier}, \binits{G.}},
\bauthor{\bsnm{Xu}, \binits{Q.}}:
\bctitle{Noncommutative $l_p$-spaces}.
In: \bbtitle{Handbook of the Geometry of Banach Spaces}
vol. \bseriesno{2},
pp. \bfpage{1459}--\blpage{1517}.
\bpublisher{North-Holland},
\blocation{Amsterdam}
(\byear{2003})
\end{bchapter}
\endbibitem

\bibitem[\protect\citeauthoryear{Xu}{2007}]{Xu07}
\begin{botherref}
\oauthor{\bsnm{Xu}, \binits{Q.}}:
Noncommutative $L_p$-spaces and martingale inequalities.
Unpublished manuscript
(2007)
\end{botherref}
\endbibitem

\bibitem[\protect\citeauthoryear{Birman et~al.}{1975}]{BKS75}
\begin{botherref}
\oauthor{\bsnm{Birman}, \binits{M.S.}},
\oauthor{\bsnm{Koplienko}, \binits{L.S.}},
\oauthor{\bsnm{Solomjak}, \binits{M.Z.}}:
Estimates of the spectrum of a difference of fractional powers of selfadjoint operators.
Izvestiya Vysshikh Uchebnykh Zavedenii. Matematika
(3),
3--10
(1975).
Issue 154
\end{botherref}
\endbibitem

\bibitem[\protect\citeauthoryear{Birman and Solomyak}{1992}]{BS89}
\begin{barticle}
\bauthor{\bsnm{Birman}, \binits{M.S.}},
\bauthor{\bsnm{Solomyak}, \binits{M.Z.}}:
\batitle{Estimates for the difference of fractional powers of selfadjoint operators under unbounded perturbations}.
\bjtitle{Journal of Soviet Mathematics}
\bvolume{61}(\bissue{2}),
\bfpage{2018}--\blpage{2035}
(\byear{1992}).
\bcomment{Translated from Zapiski Nauchnykh Seminarov Leningradskogo Otdeleniya Matematicheskogo Instituta imeni V. A. Steklova, vol. 178 (1989), pp. 120--145}
\end{barticle}
\endbibitem

\bibitem[\protect\citeauthoryear{Carlen}{2010}]{carlennote}
\begin{bchapter}
\bauthor{\bsnm{Carlen}, \binits{E.A.}}:
\bctitle{Trace inequalities and quantum entropy: An introductory course}.
In: \beditor{\bsnm{Sims}, \binits{R.}},
\beditor{\bsnm{Ueltschi}, \binits{D.}} (eds.)
\bbtitle{Entropy and the Quantum}.
\bsertitle{Contemporary Mathematics},
vol. \bseriesno{529},
pp. \bfpage{73}--\blpage{140}.
\bpublisher{American Mathematical Society},
\blocation{Providence, RI}
(\byear{2010})
\end{bchapter}
\endbibitem

\bibitem[\protect\citeauthoryear{Yu and Tan}{2018}]{YuTan2018renyi}
\begin{barticle}
\bauthor{\bsnm{Yu}, \binits{L.}},
\bauthor{\bsnm{Tan}, \binits{V.Y.}}:
\batitle{R{\'e}nyi resolvability and its applications to the wiretap channel}.
\bjtitle{IEEE Trans. Inf. Theory}
\bvolume{65}(\bissue{3}),
\bfpage{1862}--\blpage{1897}
(\byear{2018})
\end{barticle}
\endbibitem

\bibitem[\protect\citeauthoryear{Hayashi et~al.}{2025}]{hayashi2025resolvability}
\begin{barticle}
\bauthor{\bsnm{Hayashi}, \binits{M.}},
\bauthor{\bsnm{Cheng}, \binits{H.-C.}},
\bauthor{\bsnm{Gao}, \binits{L.}}:
\batitle{Resolvability of classical-quantum channels}.
\bjtitle{IEEE Transactions on Information Theory}
\bvolume{71}(\bissue{8}),
\bfpage{6061}--\blpage{6074}
(\byear{2025})
\end{barticle}
\endbibitem

\bibitem[\protect\citeauthoryear{Bergh and L{\"o}fstr{\"o}m}{1976}]{BerghLofstrom1976}
\begin{bbook}
\bauthor{\bsnm{Bergh}, \binits{J.}},
\bauthor{\bsnm{L{\"o}fstr{\"o}m}, \binits{J.}}:
\bbtitle{Interpolation Spaces: An Introduction}.
\bsertitle{Grundlehren der Mathematischen Wissenschaften},
vol. \bseriesno{223}.
\bpublisher{Springer},
\blocation{Berlin, Heidelberg}
(\byear{1976})
\end{bbook}
\endbibitem

\bibitem[\protect\citeauthoryear{Kosaki}{1984}]{kosaki1984}
\begin{barticle}
\bauthor{\bsnm{Kosaki}, \binits{H.}}:
\batitle{Applications of the complex interpolation method to a von neumann algebra: Non-commutative {$L^p$}-spaces}.
\bjtitle{Journal of Functional Analysis}
\bvolume{56}(\bissue{1}),
\bfpage{29}--\blpage{78}
(\byear{1984})
\end{barticle}
\endbibitem

\bibitem[\protect\citeauthoryear{Hirzallah and Kittaneh}{2002}]{hirzallah2002non}
\begin{barticle}
\bauthor{\bsnm{Hirzallah}, \binits{O.}},
\bauthor{\bsnm{Kittaneh}, \binits{F.}}:
\batitle{Non-commutative clarkson inequalities for unitarily invariant norms}.
\bjtitle{Pacific Journal of Mathematics}
\bvolume{202}(\bissue{2}),
\bfpage{363}--\blpage{369}
(\byear{2002})
\end{barticle}
\endbibitem

\end{thebibliography}

\end{document}